\documentclass[journal]{IEEEtran}
\usepackage{graphicx}
\usepackage{color}
\usepackage{placeins}
\usepackage{float}
\usepackage{tabularx,colortbl}
\usepackage{amssymb}
\usepackage{amsthm}
\usepackage{cite}
\usepackage{amsmath}
\usepackage{makecell}
\usepackage{subfigure}

\makeatletter                               % algorithm
\newif\if@restonecol
\makeatother

\makeatletter
\newif\if@restonecol
\makeatother

\usepackage[linesnumbered,ruled,vlined]{algorithm2e}%[ruled,vlined]{
\usepackage{algpseudocode}
\usepackage{amsmath}
\usepackage{color}
\usepackage{placeins}
\usepackage{tabularx,colortbl}
\usepackage{amssymb}

\usepackage{listings}                       % Matlab code
\usepackage[framed,numbered,autolinebreaks,useliterate]{mcode}
\theoremstyle{plain}

\newtheorem{lemm}{Lemma}

\newtheorem{coro}{Corollary}

\theoremstyle{plain}
\newtheorem{rem}{Remark}

\IEEEoverridecommandlockouts

\begin{document}

%----------------------------title&author&thanks----------------------------

\title{Low-Complexity Distributed Combining Design for Near-Field Cell-Free XL-MIMO Systems
\thanks{Z. Wang, J. Zhang, B. Xu, and B. Ai are with the State Key Laboratory of Advanced Rail Autonomous Operation, and also with the School of Electronics and Information Engineering, Beijing Jiaotong University, Beijing 100044, China, and Z. Wang is also with the Department of Communication Systems, KTH Royal Institute of Technology, 11428 Stockholm, Sweden (e-mail: \{zhewang\_77, jiayizhang, 20251197, boai\}@bjtu.edu.cn);}
\thanks{D. Niyato is with the College of Computing \& Data Science, Nanyang Technological University, Singapore 639798 (e-mail: dniyato@ntu.edu.sg);}
\thanks{S. Mao is with the Department of Electrical and Computer Engineering, Auburn University, Auburn, AL 36849-5201 USA (e-mail: smao@ieee.org);}
\thanks{Z. Han is with the Department of Electrical and Computer Engineering at the University of Houston, Houston, TX 77004 USA (e-mail: hanzhu22@gmail.com).}
}
\author{Zhe Wang, Jiayi Zhang,~\IEEEmembership{Senior Member,~IEEE,} Bokai Xu, Dusit Niyato,~\IEEEmembership{Fellow,~IEEE}, \\
Bo Ai,~\IEEEmembership{Fellow,~IEEE}, Shiwen Mao,~\IEEEmembership{Fellow,~IEEE},  Zhu Han,~\IEEEmembership{Fellow,~IEEE}}
\maketitle

%----------------------------abstract----------------------------
\begin{abstract}
In this paper, we investigate the low-complexity distributed combining scheme design for near-field cell-free extremely large-scale multiple-input-multiple-output (CF XL-MIMO) systems. Firstly, we construct the uplink spectral efficiency (SE) performance analysis framework for CF XL-MIMO systems over centralized and distributed processing schemes. Notably, we derive the centralized minimum mean-square error (CMMSE) and local minimum mean-square error (LMMSE) combining schemes over arbitrary channel estimators. Then, focusing on the CMMSE and LMMSE combining schemes, we propose five low-complexity distributed combining schemes based on the matrix approximation methodology or the symmetric successive over relaxation (SSOR) algorithm. More specifically, we propose two matrix approximation methodology-aided combining schemes: Global Statistics \& Local Instantaneous information-based MMSE (GSLI-MMSE) and Statistics matrix Inversion-based LMMSE (SI-LMMSE). These two schemes are derived by approximating the global instantaneous information in the CMMSE combining and the local instantaneous information in the LMMSE combining with the global and local statistics information by asymptotic analysis and matrix expectation approximation, respectively. Moreover, by applying the low-complexity SSOR algorithm to iteratively solve the matrix inversion in the LMMSE combining, we derive three distributed SSOR-based LMMSE combining schemes, distinguished from the applied information and initial values.

\end{abstract}
%----------------------------keywords----------------------------

%\newpage
\IEEEpeerreviewmaketitle

\section{Introduction}
With the rapid development of wireless communication, wireless communication networks have evolved to the sixth-generation (6G) \cite{you2021towards,9390169,8869705}. The 6G wireless communication networks are anticipated to satisfy the rapidly growing performance requirement and increasingly diverse application scenarios \cite{ouyang2024primer,sun2025aerial,10195219}. From the fourth generation (4G) of wireless communication networks, multiple-input multiple-output (MIMO) technology has been viewed as an important facilitator to provide excellent spectral efficiency (SE) performance of wireless networks \cite{wang2025flexible,9113273,sanguinetti2019toward}. With the continuous evolution of wireless communication, MIMO technology involves continuous growing number of antennas and more flexible hardware design. Many emerging paradigms of MIMO technology emerge to empower the future wireless communication networks, such as the extremely large-scale MIMO (XL-MIMO) \cite{ZheSurvey,9903389,CM25CKJ,zhang20236g}, cell-free massive MIMO \cite{7827017,[162],OBETrans,chen2025channel}, reconfigurable intelligent surfaces (RIS) \cite{10380596,10918608,wu2024exploit}, and movable antenna \cite{10318061,shao20246d,lipengmag,11007274}. Among these emerging MIMO technologies, XL-MIMO and CF mMIMO technologies have sparked lots of research interest for their major respective characteristics of the extremely large number of antennas \cite{ZheSurvey} and the prominent distributed network topology \cite{[162]}.

For XL-MIMO technology, an extremely large number of antennas, e.g., hundreds of even to thousands of antennas, are utilized, which are usually deployed in a dense manner with the antenna spacing much smaller than the widely applied half-wavelength antenna spacing in conventional MIMO networks \cite{ZheSurvey,TWC23CKJ}. With the involvement of these numerous antennas, one major characteristic of the XL-MIMO is the near-field spherical wave characteristic, which is different from the far-field planar wave characteristic in conventional MIMO systems \cite{2024arXiv240105900L,BokaiICC,xu2024resource,2023arXiv231011044L,11159291}. As for the near-field channel model, the authors in \cite{2023arXiv231011044L} comprehensively introduced the near-field non-uniform spherical wave (NUSW) line-of-sight channel model and mathematically compared the NUSW model with the conventional planar wave model. The authors in \cite{[41]} proposed a wavenumber domain Fourier plane-wave representation-based method to model the near-field non-line-of-sight (NLoS) near-field channel, which has been widely applied in the analysis of XL-MIMO systems \cite{[54],zhao2024performance}. Based on the spherical wave-based channel model, the authors in \cite{wang2024analytical} explored the degree-of-freedom (DoF) performance for both the discrete and continuous aperture XL-MIMO systems and found that the DoF performance for XL-MIMO systems is determined by the array aperture instead of the number of antennas. To tackle the practical issues of high hardware cost and power consumption suffered by XL-MIMO, the authors first proposed an innovative concept termed array configuration codebook (ACC) in \cite{lu2025flexible}, which provided a new promising framework to enable flexible XL-MIMO cost-effectively. In \cite{[54]}, the authors implemented performance analysis for the multi-user XL-MIMO system. However, the majority of existing works on XL-MIMO focus on the single base station (BS) scenario and it is necessary to consider practical multi-BS systems.

One promising multi-BS MIMO topology is the CF mMIMO network \cite{[162],OBETrans,chen2025channel}. In CF mMIMO networks, a large number of BSs, also called as access points (APs), are deployed at a wide area, with the number of antennas each AP relatively small, e.g. 2 or 4 \cite{[162]}. All APs embrace the capability of signal processing and are linked to one or multiple central processing units (CPUs), where all APs can cooperate with each other with the assistance of CPUs to provide uniform service for user equipments (UEs). Relying on the prominent paradigm topology, many uplink processing schemes, from the fully centralized one to the fully distributed one, can be applied \cite{[162]}. Notably, as studied in \cite{[162]}, the centralized processing and large-scale fading decoding (LSFD)-aided distributed processing schemes with their respective minimum mean-square error (MMSE)-based combining schemes are most advanced to achieve excellent SE performance. Moreover, the authors in \cite{chen2025channel} were the first to apply the novel channel map technology to the CF mMIMO multiple access scheme, where prior channel knowledge is leveraged to effectively enhance transmission performance and also holds the potential to reduce transmission complexity and overhead. Motivated by the CF mMIMO networks, the multi-BSs XL-MIMO system can also be implemented in a distributed CF manner with processing ideas like CF mMIMO networks. This paradigm is called as ``CF XL-MIMO". The basic idea of CF XL-MIMO is to implement cell-free cooperation with XL antenna apertures at a moderate number of BSs. Note that the number of antennas at each site in CF XL-MIMO is much larger while the number of sites is much smaller than those of conventional CF mMIMO. ``Cell-free" mainly emphasizes the cooperative paradigm for different sites instead of deploying a large number of sites like that in conventional CF mMIMO. Compared to the widely studied single-site XL-MIMO network, distributing the antenna array across coordinated sites in the CF XL-MIMO network can provide more uniform near-field capability and enhanced macro-diversity. Compared to the CF mMIMO network, the CF XL-MIMO network, which relies on the XL-array at each site, showcases enhanced per-site local signal processing and computing capability. This alleviates the centralized processing burden and facilitates the practical implementation of CF networks in some hot-spot scenarios, such as the stadium and campus.

Some recent works consider this CF XL-MIMO paradigm \cite{lei2023uplink,10475888,liu2023double}. The authors in \cite{lei2023uplink} investigated the uplink SE performance for CF XL-MIMO networks with the distributed processing scheme. The authors in \cite{10475888} and \cite{liu2023double} focused on the multi-agent reinforcement
learning (MARL)-empowered CF XL-MIMO networks. More specifically, the authors in \cite{10475888} studied the MARL-aided CF XL-MIMO networks, from the perspective of antenna selection and power control. The authors in \cite{liu2023double} investigated a double-layer MARL-enabled power control scheme for CF XL-MIMO networks, which can efficiently suppress the interference. However, there exists a research gap for the uplink combining scheme design for CF XL-MIMO networks. The combining design poses significant importance to efficiently decode the signal and achieve the excellent SE performance. Many works have studied the combining scheme design for CF mMIMO networks \cite{[162],OBETrans}. However, there exist the following major challenges to directly utilize the widely applied combining schemes in CF mMIMO networks. Firstly, both the prevailing centralized and distributed competitive combining scheme in CF mMIMO networks, e.g., centralized MMSE (CMMSE) as in \cite[Eq. (13)]{[162]} and local MMSE (LMMSE) combining as in \cite[Eq. (16)]{[162]}, require the prerequisite of the MMSE channel estimator. However, applying the MMSE channel estimator in XL-MIMO systems poses significant computational complexity, which restricts the practical implementation of XL-MIMO networks. Thus, it is urging to study the MMSE combining schemes under a generalized channel estimator instead of limiting the channel estimator as the MMSE one. Secondly, the computation of MMSE combining schemes involves computationally demanding instantaneous information matrix inversion, which is of significantly high computational complexity when the number of antennas is large. Therefore, it is necessary to derive some low-complexity MMSE combining schemes with the aid of some potential low-complexity methods.

Motivated by the above statements, in this paper, we study the low-complexity combining scheme design for near-field CF XL-MIMO systems. Our main contributions are listed as follows.
\begin{itemize}
\item We provide the SE performance analysis frameworks for both the centralized and distributed processing-aided CF XL-MIMO systems. Notably, the near-field channel model and the mutual coupling effect are involved. More importantly, we derive the CMMSE and LMMSE combining schemes under arbitrary channel estimators. 
\item With the aid of matrix approximation, we propose two low-complexity distributed MMSE-based combining schemes. The first one is called Global Statistics \& Local Instantaneous information-based MMSE (GSLI-MMSE), where the matrix approximation is applied to approximate the global instantaneous information. The second one is called 
Statistics matrix Inversion-based LMMSE (SI-LMMSE), where the local instantaneous information in the matrix inversion is approximated by the local statistic information, and therefore the matrix inversion for this combining is based on the statistics information.
\item By employing the low-complexity symmetric successive over relaxation (SSOR) algorithm to iteratively solve the computational demanding matrix inversion in the LMMSE combining, we propose three SSOR algorithm-based LMMSE combining schemes. The first two, Instantaneous SSOR algorithm-aided LMMSE (Ins-SSOR-LMMSE) and Statistics matrix inversion SSOR algorithm-aided LMMSE (Sta-SSOR-LMMSE), are derived by directly applying the SSOR algorithm to the LMMSE and SI-MMSE combining schemes, respectively, with zero initial value. The third one is called Instantaneous SSOR algorithm with Statistics-based Initial value aided LMMSE (Ins-SI-SSOR-LMMSE), where the statistics-based initial value is introduced to enhance the achievable performance and accelerate the convergence speed of the SSOR algorithm.
\end{itemize}

The rest of this paper is organized as follows. Section~\ref{system} introduces the system model of the considered CF XL-MIMO network including the fundamentals of the near-field channel model and the mutual coupling effect. In Section~\ref{UL}, we introduce the fundamentals of uplink transmission, where the centralized and distributed processing schemes are studied and their respective achievable SE expressions are derived. In Section~\ref{Sec_AA}, two matrix approximation-aided low-complexity distributed MMSE-based combining schemes, GSLI-MMSE and SI-LMMSE, are proposed. Section~\ref{Sec_SSOR} presents three SSOR algorithm-based LMMSE combining schemes, Ins-SSOR-LMMSE, Sta-SSOR-LMMSE, and Ins-SI-SSOR-LMMSE combining schemes. In Section~\ref{numerical}, numerical results for all considered low-complexity distributed combining schemes are presented. Finally, we draw the major conclusions in Section~\ref{conclusion}.

\textbf{\emph{Notation}}: Let boldface lowercase letters $\mathbf{x}$ and boldface uppercase letters $\mathbf{X}$ denote the column vectors and matrices, respectively. 
$\mathrm{diag}\left( \mathbf{X}_1,\cdots ,\mathbf{X}_n \right)$ represents a block-diagonal matrix, where the square matrices $\mathbf{X}_1,\cdots ,\mathbf{X}_n $ are located on the diagonal. We define $\triangleq$, $\mathbb{E} \left\{ \cdot \right\}$, and $\mathrm{tr}\left\{ \cdot \right\}$ as the definitions, the expectation operator, and the trace operator, respectively. Let $\left( \cdot \right) ^H$, $\left( \cdot \right) ^T$, and $\left( \cdot \right) ^*$ represent the conjugate transpose, transpose, and conjugate, respectively. $\mathbf{I}_{n\times n}$ denotes the $n\times n$ identity matrix. We define $\mathbf{x}\sim \mathcal{N} _{\mathbb{C}}\left( \bf{0},\mathbf{R} \right)$ as a circularly symmetric complex Gaussian distribution vector with correlation matrix $\mathbf{R}$.

\vspace*{-0.3cm}
\begin{figure}[t]
\centering
\includegraphics[scale=0.37]{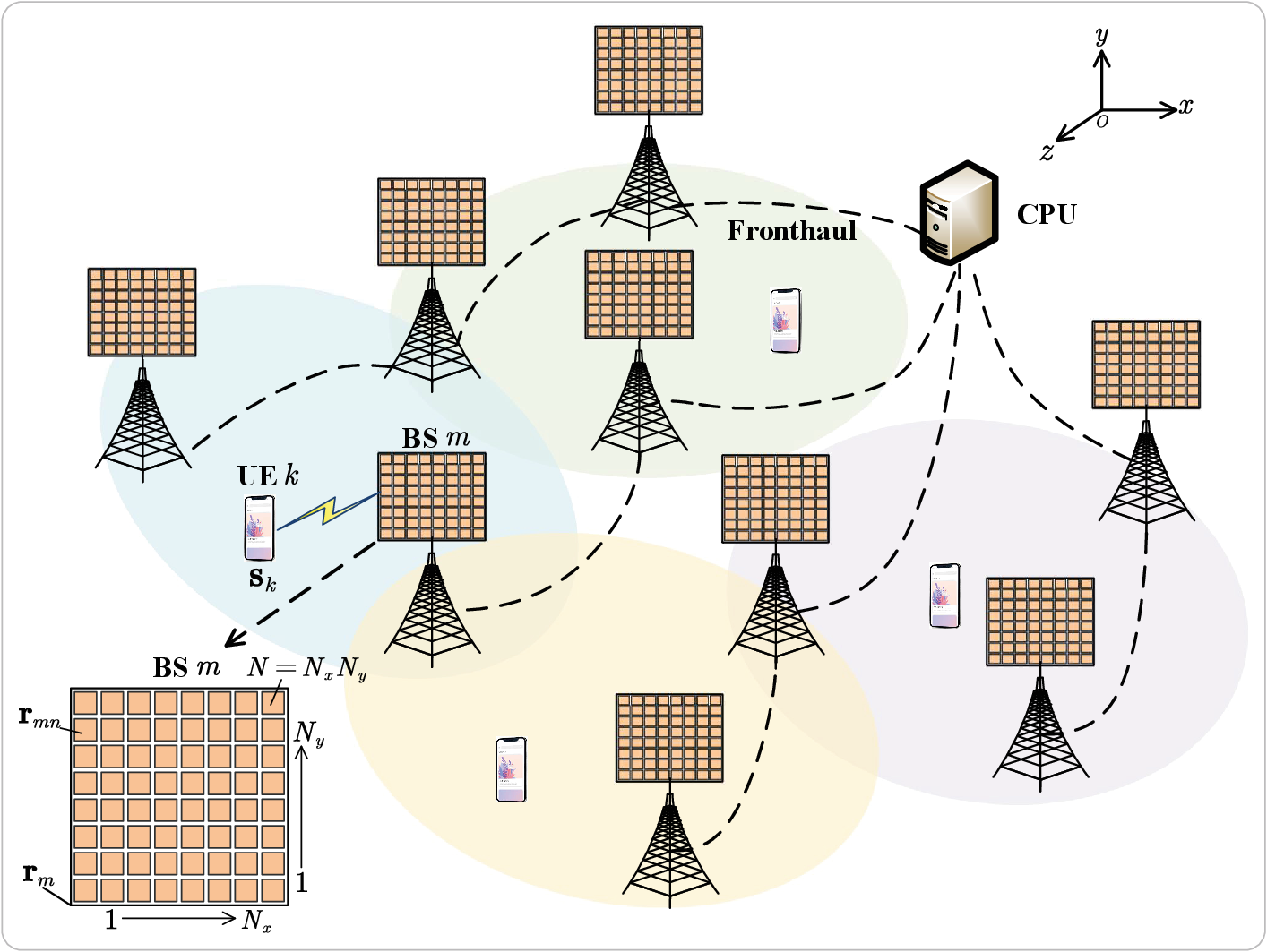}
\caption{Illustration of the studied CF XL-MIMO network. \label{System}}
\vspace{-0.2cm}
\end{figure}

\section{System Model}\label{system}
In this paper, we investigate a CF XL-MIMO network, where $M$ BSs, linked via fronthaul links, serve $K$ single-antenna UEs in a wide $L \times L  \, \mathrm{m}^2$ square coverage area, as illustrated as Fig.~\ref{System}. Each BS is equipped with the XL-UPA, where the numbers of antennas each BS in the horizontal and vertical directions are $N_x$ and $N_y$, respectively. Thus, each XL-UPA has total $N=N_x\times N_y$ antennas. We let the center of the coverage area be the origin and we assume that all UPAs are parallel to the $x-y$ plane. Then, let the coordinates of the bottom left point of the $m$-th BS be $\mathbf{r}_m=[ x_m,y_m,z_m ] ^T\in \mathbb{R} ^3$ with $y_m=L_{BS}$ being the height of each BS. Following the similar representation method as in \cite{wang2024analytical}, we can derive the position of the $n$-th antenna for the $m$-th BS, where the antenna index is counted row-by-row from the bottom left, as $\mathbf{r}_{mn}=[r_{mn,x},r_{mn,y},r_{mn,z}]^T=[ x_m+\mathrm{mod}( n-1,N_x ) \Delta _x,L_{BS}+\lfloor ( n-1 ) /N_x \rfloor \Delta _y,z_m ] ^T\in \mathbb{R} ^3$. $\Delta _x$ and $\Delta _y$ denote the horizontal and vertical antenna spacing, respectively. Moreover, we define the position of the $k$-th UE as $\mathbf{s}_k=[s_{k,x},s_{k,y},s_{k,z}] ^T\in \mathbb{R} ^3$, where $s_{k,y}=L_{UE}$ is the height of each UE.

\subsection{Near-Field Channel Model}\label{channel}
The near-field channel model with both the LoS and NLoS components is considered. We define the uplink channel vector between the $m$-th BS and the $k$-th UE as $\mathbf{h}_{mk}=\overline{\mathbf{h}}_{mk}+\check{\mathbf{h}}_{mk}\in \mathbb{C} ^N$, where $\overline{\mathbf{h}}_{mk}$ and $\check{\mathbf{h}}_{mk}$ are the LoS and NLoS components, respectively. For the LoS component $\overline{\mathbf{h}}_{mk}$, we apply the non-uniform spherical wave (NUSW) propagation model introduced in 
\cite[Sec. II]{2023arXiv231011044L}. More specifically, we have $\overline{\mathbf{h}}_{mk}=\sqrt{\beta _{mk}^{\mathrm{LoS}}}\left[ \frac{d_{mk}}{d_{m1k}}e^{-j\frac{2\pi}{\lambda}( d_{m1k}-d_{mk} )},\dots ,\frac{d_{mk}}{d_{mNk}}e^{-j\frac{2\pi}{\lambda}( d_{mNk}-d_{mk} )} \right] ^T$, 
where $d_{mnk}=| \mathbf{r}_{mn}-\mathbf{s}_k |$ is the distance between the $n$-th antenna of the $m$-th BS and the $k$-th UE, $d_{mk}=| \mathbf{r}_{m}-\mathbf{s}_k |$ is the reference distance, $\lambda$ is the wavelength, and $\beta _{mk}^{\mathrm{LoS}}$ denotes the large-scale fading coefficient for the LoS component between BS $m$ and UE $k$, which captures many impacts like the distance-dependent pathloss, transceiver hardware, and antenna gain, e.t.c. 

For the NLoS component $\check{\mathbf{h}}_{mk}$, the Fourier plane-wave representation method \cite{[41],[54],zhao2024performance} is utilized to model it. Note that the $n$-th element of $\check{\mathbf{h}}_{mk}$ can be modeled based on the four-dimensional (4D) Fourier plane-wave representation \cite{[41]} as
\begin{equation}\label{NLOS_4D}
\begin{aligned}
[ \check{\mathbf{h}}_{mk} ]_n\!\!=\!&\frac{1}{4\pi ^2}\!\!\int{\!\!\int\!\!{\int\!\!{\int_{\mathcal{D} \times \mathcal{D}}\!\!{a_r\left( k_x,k_y,\mathbf{r}_{mn} \right)H_a\left( k_x,k_y,\kappa _x,\kappa _y \right)}}}}\\
&\times a_s( \kappa _x,\kappa _y,\mathbf{s}_k ) dk_xdk_yd\kappa _xd\kappa _y,
\end{aligned}
\end{equation}
where $a_r( k_x,k_y,\mathbf{r}_{mn} ) =e^{j[ k_xr_{mn,x}+k_yr_{mn,y}+\gamma ( k_x,k_y ) r_{mn,z} ]}$ is the receiving response, $a_s( \kappa _x,\kappa _y,\mathbf{s}_k ) =e^{-j[ \kappa _xs_{k,x}+\kappa _ys_{k,y}+\gamma ( \kappa _x,\kappa _y ) s_{k,z} ]}$ is the transmitting response, $\boldsymbol{k }=[ k _x,k _y,k_z ] ^T$ is the receiving wave vector, $\boldsymbol{\kappa }=[ \kappa _x,\kappa _y,\kappa _z ] ^T$ is the transmitting wave vector, $ \gamma ( \kappa _x,\kappa _y ) =( \kappa ^2-\kappa _{x}^{2}-\kappa _{y}^{2} ) ^{{{1}/{2}}}$, and $\mathcal{D} =\{  (k_x,k_y)\in \mathbb{R} ^2 |k_{x}^{2}+k_{y}^{2}\le \kappa ^2 \} $ denotes the integration region with $\kappa=2\pi/\lambda$ being the wavenumber. Moreover, the angular response $H_a( k_x,k_y,\kappa _x,\kappa _y ) $ is modeled as 
\begin{equation}\notag
H_a\left( k_x,k_y,\kappa _x,\kappa _y \right) =\frac{A\left( k_x,k_y,\kappa _x,\kappa _y \right) W\left( k_x,k_y,\kappa _x,\kappa _y \right)}{\gamma ^{{{1}/{2}}}\left( k_x,k_y \right) \gamma ^{{{1}/{2}}}\left( \kappa _x,\kappa _y \right)},\end{equation}
where $A( k_x,k_y,\kappa _x,\kappa _y ) $ denotes the spectral factor, which is set as $A( k_x,k_y,k_z ) =\frac{2\pi}{\sqrt{\kappa}}$ in the  isotropic scattering environment, and $W( k_x,k_y,\kappa _x,\kappa _y ) \sim \mathcal{C} \mathcal{N} ( 0,1 ) $ collects the random characteristics. Furthermore, \eqref{NLOS_4D} can be approximately expanded based on the theory of Fourier plane-wave expansion \cite{[41],[54],zhao2024performance} as
\begin{equation}\label{NLOS_expansion}
\begin{aligned}
\check{\mathbf{h}}_{mk}\approx \sqrt{N}\sum_{\left( \ell _x,\ell _y \right) \in \mathcal{E} _r}{H_{a,mk}\left( \ell _x,\ell _y \right) \mathbf{a}_{mk}\left( \ell _x,\ell _y \right)},
\end{aligned}
\end{equation}
where the $n$-th element of $\mathbf{a}_{mk}$ is $[ \mathbf{a}_{mk}( \ell _x,\ell _y ) ] _n=\frac{1}{\sqrt{N}}a_r( \frac{2\pi \ell _x}{L_{r,x}},\frac{2\pi \ell _y}{L_{r,y}},\mathbf{r}_{mn} ) e^{-j\kappa s_{k,z}}$ with $L_{r,x}=N_x\Delta _x$ and $L_{r,y}=N_y\Delta _y$, and $\mathcal{E} _r=\{ ( \ell _x,\ell _y ) \in \mathbb{Z} ^2:( \ell _x\lambda /L_{r,x} ) ^2+( \ell _y\lambda /L_{r,y} ) ^2\le 1 \} $ is the lattice ellipse at the receiver. Moreover, we have $H_{a,mk}( \ell _x,\ell _y ) \sim \mathcal{C} \mathcal{N} ( 0,\beta _{mk}^{\mathrm{NLoS}}\bar{\sigma}_{mk}^{2}( \ell _x,\ell _y ) ) $, where $\beta _{mk}^{\mathrm{NLoS}}$ denotes the large-scale fading coefficient for the NLoS component between BS $m$ and UE $k$ and $\bar{\sigma}_{mk}^{2}( \ell _x,\ell _y ) $ is the normalized variance at the sampling point $( \ell _x,\ell _y )$ with $\bar{\sigma}_{mk}^{2}( \ell _x,\ell _y ) ={{\sigma}_{mk}^{2}\left( \ell _x,\ell _y \right)}/{\sum_{\left( \ell _x,\ell _y \right) \in \mathcal{E} _r}{\sigma _{mk}^{2}\left( \ell _x,\ell _y \right)}}$, where $\sigma _{mk}^{2}( \ell _x,\ell _y )$ is the variance at $( \ell _x,\ell _y )$, which can be computed as in \cite[Appendix]{[41]}.
Based on the structure of \eqref{NLOS_expansion}, we can further construct $\check{\mathbf{h}}_{mk}$ as 
\begin{equation}\label{NLOS_channel}
\begin{aligned}
\check{\mathbf{h}}_{mk}&=\sqrt{N}\!\!\sum_{( \ell _x,\ell _y ) \in \mathcal{E} _r}\!\!{H_{a,mk}( \ell _x,\ell _y ) \mathbf{a}_{mk}( \ell _x,\ell _y )}\\
&=\mathbf{U}_{mk}\mathbf{\Sigma }_{mk}^{{{1}/{2}}}\check{\mathbf{h}}_{mk,\mathrm{iid}},
\end{aligned}
\end{equation}
where $\mathbf{U}_{mk}\in \mathbb{C} ^{N\times n_r}$ is a deterministic semi-unitary matrix, collecting $n_r$ vectors $\{ \mathbf{a}_{mk}\left( \ell _x,\ell _y \right) \in \mathbb{C} ^{N} \} _{( \ell _x,\ell _y ) \in \mathcal{E} _r}$, with $\mathbf{U}_{mk}\mathbf{U}_{mk}^{H}=\mathbf{I}_{n_r}$ and $n_r=\left| \mathcal{E} _r \right|$, $\mathbf{\Sigma }_{mk}\in \mathbb{C} ^{n_r\times n_r}$ is a diagonal matrix with diagonal elements being $\{ N\beta _{mk}^{\mathrm{NLoS}}\bar{\sigma}_{mk}^{2}( \ell _x,\ell _y ) \} _{( \ell _x,\ell _y ) \in \mathcal{E} _r} $, and $\check{\mathbf{h}}_{mk,\mathrm{iid}}\sim \mathcal{C} \mathcal{N} ( \mathbf{0},\mathbf{I}_{n_r} ) $. Furthermore, we can compute the spatial correlation matrix for $\check{\mathbf{h}}_{mk}$ as $\mathbf{R}_{mk}=\mathbb{E} \{ \check{\mathbf{h}}_{mk}\check{\mathbf{h}}_{mk}^{H} \} =\mathbf{U}_{mk}\mathbf{\Sigma }_{mk}\mathbf{U}_{mk}^{H}\in \mathbb{C} ^{N\times N}$, which satisfies $\beta _{mk}^{\mathrm{NLoS}}=\mathrm{tr}( \mathbf{R}_{mk} ) /N$. In summary, we can easily derive $\mathbf{h}_{mk}$, which is composed of both the LoS and NLoS components.

\begin{rem}
The near-field spherical wave channel introduced above can reduce to the far-field planar wave channel as the transmitting distance increases. Notably, as will be discussed in Remark~\ref{general}, our studied combining schemes in this paper showcase promising applicability to many channel models instead of just the near-field-specific ones. Thus, the near-field channel model methodology is a representative choice, instead of a prerequisite for our proposed combining schemes.
\end{rem}

\begin{rem}
There are many distance boundaries to distinguish the near-field and far-field, and a well-studied one is the Rayleigh distance, which can be computed as $2D^2/\lambda$ for a single antenna array, where $D$ is the array aperture. Note that this Rayleigh distance is computed based on the maximum phase error across the antenna array \cite{sun2025differentiate}. However, this single-array motivated distance boundary may not be directly applicable to the CF XL-MIMO network since there exists more than one antenna array in a cooperative processing manner. Thus, it is not rigorous to infer a network-level near-field boundary from any one antenna array in isolation. Indeed, the network-level near-field region boundary for the CF XL-MIMO network remains a challenging open research direction. Some criteria, motivated by the antenna array gain or degree-of-freedom \cite{sun2025differentiate} across the whole network with cooperative BSs, may be the potential solutions.
\end{rem}

\subsection{Mutual Coupling Effect}\label{Mutual}
In XL-MIMO systems, the number of antennas is extremely large. To deploy these numerous antennas, the antenna spacing is typically much smaller than the conventional half-wavelength antenna spacing in massive MIMO systems \cite{ZheSurvey}. When the antenna spacing is small, the mutual coupling effect would be severe. In this part, we introduce the modeling of the mutual coupling effect. We define the mutual coupling matrix at each BS as $\mathbf{Z}_{BS}\in \mathbb{C} ^{N\times N}$, which can be denoted as \cite{balanis2016antenna,li2017exploiting}
\begin{equation}\label{MC_Matrix}
\mathbf{Z}_{BS}=( Z_A+Z_L ) ( \mathbf{Z}_{BS,C}+Z_L\mathbf{I}_N ) ^{-1},
\end{equation}
where $ Z_A$, $ Z_L$, and $\mathbf{Z}_{BS,C}\in \mathbb{C} ^{N\times N}$ denote the antenna impedance, load impedance, and mutual impedance matrix at the BS side, respectively. We can further denote $ Z_A$ as $Z_A=R_{Z_A}+jX_{Z_A}$ with $R_{Z_A}$ and $X_{Z_A}$ being the resistance and the reactance, respectively. Moreover, based on \cite{li2017exploiting}, we have $R_{Z_A}=\frac{\eta}{2\pi \sin ^2( \frac{\kappa \Delta _l}{2} )}\{ \gamma _0+\ln ( \kappa \Delta _l ) -C_i( \kappa \Delta _l ) +\frac{\sin ( \kappa \Delta _l )}{2}[ S_i( 2\kappa \Delta _l ) -2S_i( \kappa \Delta _l ) ] +\frac{\cos ( \kappa \Delta _l )}{2}[ \gamma _0+\ln ( \frac{\kappa \Delta _l}{2} ) +C_i( 2\kappa \Delta _l ) -2C_i( \kappa \Delta _l ) ] \} $ and 
$X_{Z_A}=\frac{\eta}{4\pi \sin ^2( \frac{\kappa \Delta _l}{2} )}\{2S_i( \kappa \Delta _l ) +\cos ( \kappa \Delta _l ) [ 2S_i( \kappa \Delta _l ) -S_i( 2\kappa \Delta _l ) ]-\sin ( \kappa \Delta _l ) [ 2C_i( \kappa \Delta _l ) -C_i( 2\kappa \Delta _l ) -C_i( \frac{2\kappa r^2}{\Delta _l} ) ] \label{XZA} \}$, where $\Delta _l$ is the dipole length of each antenna element, $r$ denotes the wire radius, $\eta $ denotes the intrinsic impedance, $\gamma _0$ is the Euler constant, respectively. Besides, $S_i(\cdot)$ and $C_i(\cdot)$ denote the sine and cosine integral functions, respectively. The modeling of $\mathbf{Z}_{BS,C}$ is given in Appendix~\ref{ZBS}.

\section{Uplink Transmission}\label{UL}
Relying on the distributed deployment as in Fig.~\ref{System}, each BS embraces the ability to receive and process the signals. Mainly two uplink processing schemes can be implemented in the studied CF XL-MIMO system: centralized and distributed processing schemes. These two schemes are motivated from the promising CF mMIMO technology \cite{[162],OBETrans}. To maintain the overall structure of this paper, we first introduce the phase of channel estimation, and then briefly introduce these two processing schemes.

\subsection{Pilot Transmission and Channel Estimation}\label{CE}
In the phase of pilot transmission, $\tau _p$ mutually orthogonal pilot signals, $\boldsymbol{\phi} _1,\dots ,\boldsymbol{\phi} _{\tau _p}$, are applied\footnote{This orthogonal pilot setting poses challenges in scenarios with high user density, which showcase the application of non-orthogonal pilot \cite{9169906} and robust signal processing schemes under severe pilot contamination.}, where $\left\| \boldsymbol{\phi} _t \right\| ^2=\tau _p$. All UEs send their respective pilot signals to all BSs and the received pilot signal at BS $m$ can be denoted as 
\begin{equation}
\mathbf{Y}_{m}^{p}=\sum_{k=1}^K{\sqrt{p_k}\mathbf{Z}_{BS}\mathbf{h}_{mk}\boldsymbol{\phi }_{t_k}^{T}}+\mathbf{N}_{m}^{p}\in \mathbb{C} ^{N\times \tau _p},
\end{equation}
where $\boldsymbol{\phi} _{t_{k}}$ denotes the pilot signal transmitted by UE $k$ with $t_k\in\{1,\dots,\tau _p\}$, $p_k$ is the transmitting power for UE $k$, and $\mathbf{N}_{m}^{p}\in \mathbb{C} ^{N\times \tau _p}$ is the received noise at BS $m$ with independent $\mathcal{N} _{\mathbb{C}}\left( 0,\sigma ^2\right)$ elements with $\sigma ^2$ being the uplink noise power. With the introduction of the mutual coupling effect, we define the effective channel between BS $m$ and UE $k$ as 
\begin{equation}
\mathbf{g}_{mk}\triangleq \mathbf{Z}_{BS}\mathbf{h}_{mk}=\mathbf{Z}_{BS}\bar{\mathbf{h}}_{mk}+\mathbf{Z}_{BS}\check{\mathbf{h}}_{mk}\triangleq \bar{\mathbf{g}}_{mk}+\check{\mathbf{g}}_{mk}\in \mathbb{C} ^N
\end{equation} 
with $\check{\mathbf{R}}_{mk}\triangleq \mathbb{E} \{ \check{\mathbf{g}}_{mk}\check{\mathbf{g}}_{mk}^{H} \} =\mathbf{Z}_{BS}\mathbf{R}_{mk}\mathbf{Z}_{BS}^{H}$ being the effective spatial correlation matrix for the effective NLoS component $\check{\mathbf{g}}_{mk}$ between BS $m$ and UE $k$.

With the aid of the theories in \cite[Sec. II-A]{[162]}, we can easily derive the MMSE estimate of $\mathbf{g}_{mk}$ as 
\begin{equation}\label{MMSE_Estimator}
\hat{\mathbf{g}}_{mk}^{\mathrm{MMSE}}=\bar{\mathbf{g}}_{mk}+\sqrt{p_k}\check{\mathbf{R}}_{mk}\mathbf{\Psi }_{mk}^{-1}( \mathbf{y}_{mk}^{p}-\bar{\mathbf{y}}_{mk}^{p} ),\end{equation} 
where $\mathbf{\Psi }_{mk}\triangleq{\mathbb{E} \{ ( \mathbf{y}_{mk}^{p}-\bar{\mathbf{y}}_{mk}^{p} ) \left( \mathbf{y}_{mk}^{p}-\bar{\mathbf{y}}_{mk}^{p} \right) ^H \}}/{\tau _p}=\sum_{l\in \mathcal{P} _k}{p_l\tau _p\check{\mathbf{R}}_{ml}}+\sigma ^2\mathbf{I}_N$, $\mathbf{y}_{mk}^{p}=\mathbf{Y}_{m}^{p}\boldsymbol{\phi} _{t_k}^{*}=\sum_{l\in \mathcal{P} _k}{\sqrt{p_l}\tau _p\mathbf{g}_{ml}}+\mathbf{n}_{mk}^{p}$, $\bar{\mathbf{y}}_{mk}^{p}=\sum_{l\in \mathcal{P} _k}{\sqrt{p_l}\tau _p\bar{\mathbf{g}}_{ml}}$, and $\mathbf{n}_{mk}^{p}=\mathbf{N}_{m}^{p}\boldsymbol{\phi} _{k}^{*}\sim \mathcal{N} _{\mathbb{C}}( \mathbf{0},\tau _p\sigma ^2\mathbf{I}_N ) $, with $\mathcal{P} _k$ being the subset of UEs using the same pilot as UE $k$ including itself. Note that the computation of the MMSE estimate involves the $N\times N$ matrix inversion, which is computationally intensive when $N$ is large. Motivated by the structure of the MMSE estimate, by replacing $\check{\mathbf{R}}_{mk}\mathbf{\Psi }_{mk}^{-1}$ in $\hat{\mathbf{g}}_{mk}^{\mathrm{MMSE}}$ by arbitrary channel statistics-based matrix $\mathbf{A}_{mk}$, we derive a generalized estimator as
\begin{equation}\label{generalized_CE}
\hat{\mathbf{g}}_{mk}=\bar{\mathbf{g}}_{mk}+\mathbf{A}_{mk}( \mathbf{y}_{mk}^{p}-\bar{\mathbf{y}}_{mk}^{p} ).
\end{equation}
Then, we define the channel estimation error as $\tilde{\mathbf{g}}_{mk}\triangleq\mathbf{g}_{mk}-\hat{\mathbf{g}}_{mk}$. Note that, for the MMSE estimator, the channel estimate $\hat{\mathbf{g}}_{mk}$ and estimation error $\tilde{\mathbf{g}}_{mk}$ are independent \cite{[162]}. However, for an arbitrary matrix $\mathbf{A}_{mk}$, the channel estimate and estimation error may not be independent but correlated. More specifically, based on \eqref{generalized_CE}, we can derive the covariance matrices for $\hat{\mathbf{g}}_{mk}$ and $\tilde{\mathbf{g}}_{mk}$, and the correlation matrix between $\hat{\mathbf{g}}_{mk}$ and $\tilde{\mathbf{g}}_{mk}$ as $\bar{\mathbf{R}}_{mk}\triangleq \mathbb{E} \{ \hat{\mathbf{g}}_{mk}\hat{\mathbf{g}}_{mk}^{H} \} =\bar{\mathbf{g}}_{ml}\bar{\mathbf{g}}_{ml}^{H}+\hat{\mathbf{R}}_{ml}$  with $\hat{\mathbf{R}}_{ml}= \mathbb{E} \{\mathbf{A}_{mk}( \mathbf{y}_{mk}^{p}-\bar{\mathbf{y}}_{mk}^{p})( \mathbf{y}_{mk}^{p}-\bar{\mathbf{y}}_{mk}^{p})^H\mathbf{A}_{mk}^H\} =\tau _p\mathbf{A}_{mk}\mathbf{\Psi }_{mk}\mathbf{A}_{mk}^{H}$, $\mathbf{C}_{mk}\triangleq \mathbb{E} \{ \tilde{\mathbf{g}}_{mk}\tilde{\mathbf{g}}_{mk}^{H} \} =\check{\mathbf{R}}_{mk}-\sqrt{p_k}\tau _p\check{\mathbf{R}}_{mk}\mathbf{A}_{mk}^{H}-\mathbf{A}_{mk}\sqrt{p_k}\tau _p\check{\mathbf{R}}_{mk}+\tau _p\mathbf{A}_{mk}\mathbf{\Psi }_{mk}\mathbf{A}_{mk}^{H}$, and $\mathbf{B}_{mk}\triangleq \mathbb{E} \{ \hat{\mathbf{g}}_{mk}\tilde{\mathbf{g}}_{mk}^{H} \} =\sqrt{p_k}\tau _p\mathbf{A}_{mk}\check{\mathbf{R}}_{mk}-\tau _p\mathbf{A}_{mk}\mathbf{\Psi }_{mk}\mathbf{A}_{mk}^{H}$, respectively. These matrices would also be applied in the combining design in the following. Note that the MMSE channel estimator in \eqref{MMSE_Estimator} is a special case of the generalized channel estimator introduced in \eqref{generalized_CE}. By letting $\mathbf{A}_{mk}=\sqrt{p_k}\check{\mathbf{R}}_{mk}\mathbf{\Psi }_{mk}^{-1}$, the generalized channel estimator in \eqref{generalized_CE} can reduce to the MMSE channel estimator in \eqref{MMSE_Estimator}. If the system cannot fully capture the accurate channel covariance matrices information, or if there exists severe estimation errors for channel covariance matrices, it is not feasible to implement the MMSE channel estimator based on accurate channel matrices. Another potential channel estimator choice is the element-wise MMSE (EW-MMSE) channel estimator $\mathbf{A}_{mk}=\sqrt{p_k}\mathbf{D}_{mk}\mathbf{\Gamma }_{mk}^{-1}$, where only the diagonal components of the channel matrices are applied, with $\mathbf{D}_{mk}=\mathrm{diag}( [ \check{\mathbf{R}}_{mk} ] _{nn}: n=1,\dots ,N )$ and $\mathbf{\Gamma }_{mk}=\mathrm{diag}( [ \mathbf{\Psi }_{mk} ] _{nn}: n=1,\dots ,N )$.

\subsection{Data Transmission}\label{UL_Data}
During the phase of the uplink data transmission, all UEs simultaneously transmit $\tau_c-\tau_p$ data symbols within each coherence block and the received data signal at BS $m$ can be represented as \begin{equation}\label{received_data}
\mathbf{y}_{m}^{d}=\sum_{k=1}^K{\mathbf{Z}_{BS}\mathbf{h}_{mk}x_k}+\mathbf{n}_m=\sum_{k=1}^K{\mathbf{g}_{mk}x_k}+\mathbf{n}_m,
\end{equation}
where $x_k\sim \mathcal{N} _{\mathbb{C}}( 0,p_k) $ represents the data symbol sent by UE $k$ and $\mathbf{n}_m\sim \mathcal{N} _{\mathbb{C}}( \mathbf{0},\sigma ^2\mathbf{I}_N ) $ denotes the receive noise at AP $m$. Note that after receiving the data signal as \eqref{received_data}, each BS can just serve as a relay or can implement a distributed data detection, which corresponds to the centralized or distributed processing schemes in the following.

\subsubsection{Centralized Processing}\label{centralized}
Under the centralized processing, all the received pilot and data signals at all BSs are transmitted to the CPU, where all the tasks of channel estimation and data detection are implemented at the CPU. Relying on the theories in \cite{[162]}, we can derive the estimate of the collective effective channel $\mathbf{g}_k=[ \mathbf{g}_{1k}^{T},\dots ,\mathbf{g}_{Mk}^{T} ] ^T=\bar{\mathbf{g}}_k+\check{\mathbf{g}}_k\in \mathbb{C} ^{MN}$ with $\check{\mathbf{R}}_k\triangleq \mathbb{E} \{ \check{\mathbf{g}}_k\check{\mathbf{g}}_{k}^{H} \} =\mathrm{diag}( \check{\mathbf{R}}_{1k},\dots ,\check{\mathbf{R}}_{Mk} ) \in \mathbb{C} ^{MN\times MN}$ as 
\begin{equation}\label{CE_Collective}
\hat{\mathbf{g}}_k=\bar{\mathbf{g}}_k+\mathbf{A}_k( \mathbf{y}_{k}^{p}-\bar{\mathbf{y}}_{k}^{p} ),
\end{equation}
where $\mathbf{A}_k\triangleq \mathrm{diag}( \mathbf{A}_{1k},\dots ,\mathbf{A}_{Mk} ) \in \mathbb{C} ^{MN\times MN}$, $\mathbf{y}_{k}^{p}\triangleq[ \mathbf{y}_{1k}^{p,T},\dots ,\mathbf{y}_{Mk}^{p,T} ] ^T\in \mathbb{C} ^{MN}$, $\bar{\mathbf{y}}_{k}^{p}\triangleq[ \bar{\mathbf{y}}_{1k}^{p,T},\dots ,\bar{\mathbf{y}}_{Mk}^{p,T} ] \in \mathbb{C} ^{MN}$, $\mathbf{\Psi }_k\triangleq\mathrm{diag}( \mathbf{\Psi }_{1k},\dots ,\mathbf{\Psi }_{Mk} ) \in \mathbb{C} ^{MN\times MN}$. Moreover, we have $\hat{\mathbf{R}}_{k}\triangleq \mathbb{E} \{ \mathbf{A}_k( \mathbf{y}_{k}^{p}-\bar{\mathbf{y}}_{k}^{p})( \mathbf{y}_{k}^{p}-\bar{\mathbf{y}}_{k}^{p})^H \mathbf{A}_k^H\} =\tau _p\mathbf{A}_{k}\mathbf{\Psi }_{k}\mathbf{A}_{k}^{H}=\mathrm{diag}( \hat{\mathbf{R}}_{1k},\dots ,\hat{\mathbf{R}}_{Mk} )\in \mathbb{C} ^{MN\times MN}$, $\mathbf{B}_k\triangleq \mathbb{E} \{ \hat{\mathbf{g}}_k\tilde{\mathbf{g}}_{k}^{H} \} =\sqrt{p_k}\tau _p\mathbf{A}_k\check{\mathbf{R}}_k-\tau _p\mathbf{A}_k\mathbf{\Psi }_k\mathbf{A}_{k}^{H}=\mathrm{diag}( \mathbf{B}_{1k},\dots ,\mathbf{B}_{Mk} ) \in \mathbb{C} ^{MN\times MN}$, and $\mathbf{C}_{k}\triangleq \mathbb{E} \{ \tilde{\mathbf{g}}_{k}\tilde{\mathbf{g}}_{k}^{H} \} =\check{\mathbf{R}}_{k}-\sqrt{p_k}\tau _p\check{\mathbf{R}}_{k}\mathbf{A}_{k}^{H}-\mathbf{A}_{k}\sqrt{p_k}\tau _p\check{\mathbf{R}}_{k}+\tau _p\mathbf{A}_{k}\mathbf{\Psi }_{k}\mathbf{A}_{m}^{H}=\mathrm{diag}( \mathbf{C}_{1k},\dots ,\mathbf{C}_{Mk} ) \in \mathbb{C} ^{MN\times MN}$, respectively, with $\tilde{\mathbf{g}}_k\triangleq \mathbf{g}_k-\hat{\mathbf{g}}_k$ being the estimation error of $\mathbf{g}_k$.

Furthermore, the collective received data signal can be denoted as $\mathbf{y}^d=\sum_{k=1}^K{\mathbf{g}_kx_k}+\mathbf{n}$ with $\mathbf{n}\sim \mathcal{N} _{\mathbb{C}}( \mathbf{0},\sigma ^2\mathbf{I}_{MN}) $. The CPU can select an arbitrary received combining scheme $\mathbf{v}_k\in \mathbb{C} ^{MN}$ to decode the data symbol for UE $k$ as
\begin{equation}
\hat{x}_k=\mathbf{v}_{k}^{H}\mathbf{y}^d=\mathbf{v}_{k}^{H}\mathbf{g}_kx_k+\sum_{l\ne k}^K{\mathbf{v}_{k}^{H}\mathbf{g}_lx_l}+\mathbf{v}_{k}^{H}\mathbf{n}.
\end{equation}
One well-studied combining scheme is the MMSE combining, which can minimize the conditional MSE $\mathrm{MSE}_k=\mathbb{E} \{ | x_k-\hat{x}_k |^2 |\hat{\mathbf{g}}_k \} $. However, most of existing MMSE combining formulations, such as \cite[Eq. (13)]{[162]} and \cite[Eq. (12)]{OBETrans}, are based on the MMSE channel estimator, which mismatch the scenario with the generalized estimator in this paper. Thus, it is vital to derive the generalized estimator-based MMSE combining scheme as in the following corollary.

\begin{coro}\label{C-MMSE}
Under the generalized estimator as in \eqref{CE_Collective}, the centralized MMSE (CMMSE) combining scheme minimizing $\mathrm{MSE}_k=\mathbb{E} \{ | x_k-\hat{x}_k |^2 |\hat{\mathbf{g}}_k \} $ can be represented as 
\begin{equation}\label{eq_CMMSE}
\mathbf{v}_k=p_k\left[ \sum_{l=1}^K{p_l\left( \hat{\mathbf{g}}_l\hat{\mathbf{g}}_{l}^{H}+\mathbf{B}_l+\mathbf{B}_{l}^{H}+\mathbf{C}_l \right)}+\sigma ^2\mathbf{I}_{MN} \right] ^{-1}\hat{\mathbf{g}}_k.
\end{equation}
\end{coro}
\begin{IEEEproof}
The conditional MSE can be expanded as $\mathrm{MSE}_k=\mathbb{E} \{ | x_k-\hat{x}_k |^2 |\hat{\mathbf{g}}_k \}=\mathbb{E} \{ x_kx_{k}^{H} |\hat{\mathbf{g}}_k \} -\mathbb{E} \{ x_k\hat{x}_{k}^{H} |\hat{\mathbf{g}}_k \} -\mathbb{E} \{  \hat{x}_kx_{k}^{H} |\hat{\mathbf{g}}_k \} +\mathbb{E} \{  \hat{x}_k\hat{x}_{k}^{H} |\hat{\mathbf{g}}_k \} 
=p_k-p_k\hat{\mathbf{g}}_{k}^{H}\mathbf{v}_k-p_k\mathbf{v}_{k}^{H}\hat{\mathbf{g}}_k
+\mathbf{v}_{k}^{H}[ \sum_{l=1}^K{( p_l\hat{\mathbf{g}}_l\hat{\mathbf{g}}_{l}^{H}+p_l\mathbf{B}_l+p_l\mathbf{B}_{l}^{H}+p_l\mathbf{C}_l )}+\sigma ^2\mathbf{I}_{MN} ] \mathbf{v}_k$, where $\mathbb{E} \{  \hat{x}_k\hat{x}_{k}^{H} |\hat{\mathbf{g}}_k \} =\mathbf{v}_{k}^{H}( \mathbb{E} \{  \sum_{l=1}^K{p_l( \hat{\mathbf{g}}_l+\tilde{\mathbf{g}}_l ) ( \hat{\mathbf{g}}_l+\tilde{\mathbf{g}}_l ) ^H} |\hat{\mathbf{g}}_k \} +\sigma ^2\mathbf{I}_{MN} ) \mathbf{v}_k$ with $\mathbb{E} \{  \sum_{l=1}^K{p_l( \hat{\mathbf{g}}_l+\tilde{\mathbf{g}}_l ) ( \hat{\mathbf{g}}_l+\tilde{\mathbf{g}}_l ) ^H} |\hat{\mathbf{g}}_k \} =\mathbb{E} \{  \sum_{l=1}^K{p_l\hat{\mathbf{g}}_l\hat{\mathbf{g}}_{l}^{H}} |\hat{\mathbf{g}}_k \} +\mathbb{E} \{ \sum_{l=1}^K{p_l\hat{\mathbf{g}}_l\tilde{\mathbf{g}}_{l}^{H}} |\hat{\mathbf{g}}_k \} +\mathbb{E} \{  \sum_{l=1}^K{p_l\tilde{\mathbf{g}}_l\hat{\mathbf{g}}_{l}^{H}} |\hat{\mathbf{g}}_k \} +\mathbb{E} \{  \sum_{l=1}^K{p_l\tilde{\mathbf{g}}_l\tilde{\mathbf{g}}_{l}^{H}} |\hat{\mathbf{g}}_k \} =\sum_{l=1}^K{( p_l\hat{\mathbf{g}}_l\hat{\mathbf{g}}_{l}^{H}+p_l\mathbf{B}_l+p_l\mathbf{B}_{l}^{H}+p_l\mathbf{C}_l )}$ based on the results derived above. Then, by letting $\frac{\partial \mathrm{MSE}_k}{\mathbf{v}_k}=0$, we can easily derive the MMSE combining as \eqref{eq_CMMSE} that minimizes $ \mathrm{MSE}_k$.
\end{IEEEproof}

Furthermore, we can evaluate the achievable SE performance for the centralized processing scheme. One well-studied evaluation method is to compute the achievable SE by the use-and-then-forget (UatF) capacity bound \cite{8187178} as
\begin{equation}
\mathrm{SE}_{k}^{\mathrm{c},\mathrm{UatF}}=\frac{\tau _c-\tau _p}{\tau _c}\log _2( 1+\gamma_{k}^{\mathrm{c},\mathrm{UatF}} ), 
\end{equation}
where the signal-to-interference-and-noise ratio (SINR) $\gamma_{k}^{\mathrm{c},\mathrm{UatF}}$ is
\begin{equation}\label{SINR_centrlized_UatF}
\gamma_{k}^{\mathrm{c},\mathrm{UatF}}\!\!=\!\!\frac{p_k| \mathbb{E} \{ \mathbf{v}_{k}^{H}\mathbf{g}_k \} |^2}{\sum\limits_{l=1}^K{p_l\mathbb{E} \{ | \mathbf{v}_{k}^{H}\mathbf{g}_l |^2 \}}\!-\!p_k| \mathbb{E} \{ \mathbf{v}_{k}^{H}\mathbf{g}_k \} |^2\!+\!\sigma ^2\mathbb{E} \{ \| \mathbf{v}_k \| ^2 \}}.
\end{equation}
Note that this capacity bound holds for any channel estimators. When the MMSE estimator is applied, the channel estimate and estimation error are independent, that is $\mathbf{B}_k=\mathbf{0}$. One more tight standard capacity bound \cite{8187178} can be applied to evaluate the achievable SE for UE $k$ as 
\begin{equation}
\mathrm{SE}_{k}^{\mathrm{c},\mathrm{stan}}=\frac{\tau _c-\tau _p}{\tau _c}\mathbb{E} \{ \log _2( 1+\gamma_{k}^{\mathrm{c},\mathrm{stan}} ) \} \end{equation}
with 
\begin{equation}\label{SINR_centrlized_Standard}
\begin{aligned}
\gamma_{k}^{\mathrm{c},\mathrm{stan}}=\frac{p_k\left| \mathbf{v}_{k}^{H}\hat{\mathbf{g}}_k \right|^2}{\sum\limits_{l\ne k}^K{p_l| \mathbf{v}_{k}^{H}\hat{\mathbf{g}}_l |^2}+\mathbf{v}_{k}^{H}( \sum\limits_{l=1}^K{p_l\mathbf{C}_l}+\sigma ^2\mathbf{I}_{MN} ) ^{-1}\mathbf{v}_k}.
\end{aligned}
\end{equation}
These two capacity bounds are both well applied in the CF mMIMO research. The standard capacity bound is computed based on the instantaneous information, with the prerequisite that the MMSE channel estimator is applied. On the contrary, the UatF capacity bound holds for any type of channel estimator and is computed based on the statistical information. The performance gap between these two bounds is very small, which has been evaluated in \cite[Fig. 2]{OBETrans}.

\subsubsection{Distributed Processing}\label{distributed}
Each BS showcases the capability to process the signals with the aid of multiple antennas. In this subsection, we introduce a distributed processing scheme, where the channel estimation and first layer soft data detection are implemented at each BS, and the second layer final data detection is applied at the CPU. This distributed processing scheme is motivated from CF mMIMO networks \cite{[162]}. 

More specifically, with the aid of local channel estimates, each BS can design a local combining scheme for UE $k$ $\mathbf{v}_{mk}\in \mathbb{C} ^N$ to locally implement a soft data detection as 
\begin{equation}
\begin{aligned}
\hat{x}_{mk}=\mathbf{v}_{mk}^{H}\mathbf{y}_{m}^d=\mathbf{v}_{mk}^{H}\mathbf{g}_{mk}x_k+\sum_{l\ne k}^K{\mathbf{v}_{mk}^{H}\mathbf{g}_{ml}x_l}+\mathbf{v}_{mk}^{H}\mathbf{n}_m.
\end{aligned}
\end{equation}
Many distributed combining schemes can be applied \cite{[162],OBETrans}. One widely applied combining scheme is the LMMSE scheme as \cite[Eq. (16)]{[162]}, which can minimize the local conditional MSE $\mathrm{MSE}_{mk}=\mathbb{E} \{  | x_k-\hat{x}_{mk} |^2 |\hat{\mathbf{g}}_{mk} \} $. Note that well-studied LMMSE combining formulation including \cite[Eq. (16)]{[162]} also holds only for the MMSE channel estimator. We can derive the LMMSE combining scheme, which holds for arbitrary channel estimators, as in the following corollary.
\begin{coro}\label{L-MMSE}
With the generalized channel estimator as in \eqref{generalized_CE}, the LMMSE combining scheme, which minimizes $\mathrm{MSE}_{mk}=\mathbb{E} \{  | x_k-\hat{x}_{mk} |^2 |\hat{\mathbf{g}}_{mk} \} $ is given by
\begin{equation}\label{eq_LMMSE}
\mathbf{v}_{mk}\!\!=\!\!p_k\!\!\left[ \sum_{l=1}^K{p_l( \hat{\mathbf{g}}_{ml}\hat{\mathbf{g}}_{ml}^{H}+\mathbf{B}_{ml}+\mathbf{B}_{ml}^{H}+\mathbf{C}_{ml} )}+\sigma ^2\mathbf{I}_N \right] ^{-1}\!\!\!\hat{\mathbf{g}}_{mk}.
\end{equation}
\end{coro}
\begin{IEEEproof}
This result can be easily proved based on the method in Corollary~\ref{C-MMSE}.
\end{IEEEproof}
Then, all BSs send their local data estimates to the CPU, where the final decoding can be implemented at the CPU with the aid of the LSFD weight coefficients  $\{ a_{mk}:m=1,\dots ,M \} $ as $\check{x}_k=\sum_{m=1}^M{a_{mk}^{*}\hat{x}_{mk}}$. And we can compute the achievable SE for UE $k$ based on the UatF capacity bound as
\begin{equation}
\mathrm{SE}_{k}^{\mathrm{d}}=\frac{\tau _c-\tau _p}{\tau _c}\log _2(1+\gamma_{k}^{\mathrm{d}}),
\end{equation}
with
\begin{equation}
\begin{aligned}\label{SINR_LSFD}
&\gamma_{k}^{\mathrm{d}}=\\
&\frac{p_k\left| \mathbf{a}_{k}^{H}\mathbb{E} \left\{ \mathbf{b}_{kk} \right\} \right|^2}{\mathbf{a}_{k}^{H}\!(\sum_{l=1}^K{p_l\mathbb{E} \{ \mathbf{b}_{kl}\mathbf{b}_{kl}^{H} \}}\!-\!p_k\mathbb{E} \left\{ \mathbf{b}_{kk} \right\} \mathbb{E} \left\{ \mathbf{b}_{kk}^{H} \right\} +\sigma ^2\mathbf{D}_k)\mathbf{a}_k},
\end{aligned}
\end{equation}
where $\mathbf{b}_{kl}=[ \mathbf{v}_{1k}^{H}\mathbf{g}_{1l},\dots ,\mathbf{v}_{Mk}^{H}\mathbf{g}_{Ml} ] ^T\in \mathbb{C} ^M$, $\mathbf{a}_k=[ a_{1k},\dots ,a_{Mk} ] \in \mathbb{C} ^M$, and $\mathbf{D}_k=\mathrm{diag}[ \mathbb{E} \{ \| \mathbf{v}_{1k} \| ^2 \} ,\dots ,\mathbb{E} \{ \| \mathbf{v}_{Mk} \| ^2 \} ] \in \mathbb{C} ^{M\times M}$.
More importantly, \eqref{SINR_LSFD} can be maximized by 
$\mathbf{a}_{k}^{*}=( \sum_{l=1}^K{p_l\mathbb{E} \{ \mathbf{b}_{kl}\mathbf{b}_{kl}^{H} \}}-p_k\mathbb{E} \{ \mathbf{b}_{kk} \} \mathbb{E} \{ \mathbf{b}_{kk}^{H}\} +\sigma ^2\mathbf{D}_k ) ^{-1}\mathbb{E} \{ \mathbf{b}_{kk} \} $ with its maximum value $\gamma_{k}^{\mathrm{d},*}=p_k\mathbb{E} \{ \mathbf{b}_{kk}^{H} \} \mathbf{a}_{k}^{*}$ based on \cite[Corollary 2]{[162]}.

\section{Matrix Approximation Aid Distributed Combining Design}\label{Sec_AA}
In our studied CF XL-MIMO networks, the BS is usually equipped with a large number of antennas. Thus, some matrix approximation methodologies can be implemented to approximate some components in MMSE-based combining schemes introduced above to derive some efficient low-complexity distributed combining schemes. In this section, we introduce two matrix approximation aid distributed combining schemes.

\subsection{GSLI-MMSE Combining Scheme}
Firstly, we study a global statistics \& local instantaneous information (GSLI)-based MMSE combining scheme. We start from the CMMSE combining as in \eqref{eq_CMMSE}. Note that \eqref{eq_CMMSE} can be further constructed as 
\begin{equation}\label{CMMSE_Construct}
\mathbf{v}_k=p_k\left( \hat{\mathbf{G}}\mathbf{P}\hat{\mathbf{G}}^H+\mathbf{Q} \right) ^{-1}\hat{\mathbf{G}}\mathbf{e}_k,
\end{equation}
where $\hat{\mathbf{G}}\triangleq [ \hat{\mathbf{g}}_1,\dots ,\hat{\mathbf{g}}_K ] \in \mathbb{C} ^{MN\times K}$, $\mathbf{Q}=\sum_{l=1}^K{p_l( \mathbf{B}_l+\mathbf{B}_{l}^{H}+\mathbf{C}_l )}+\sigma ^2\mathbf{I}_{MN}$, $\mathbf{P}=\mathrm{diag}( p_1,\dots ,p_K ) \in \mathbb{C} ^{K\times K}$, and $\mathbf{e}_k$ is the $k$-th column of $\mathbf{I}_K$. Following the matrix inversion method in Lemma~\ref{matrixinversion}, we can further formulate \eqref{CMMSE_Construct} as 
\begin{equation}\label{CMMSE_Construct_inverse}
\mathbf{v}_k=p_k\mathbf{Q}^{-1}\hat{\mathbf{G}}\left( \hat{\mathbf{G}}^H\mathbf{Q}^{-1}\hat{\mathbf{G}}+\mathbf{P}^{-1} \right) ^{-1}\mathbf{P}^{-1}\mathbf{e}_k,
\end{equation}
by letting the matrices in Lemma~\ref{matrixinversion} be $\mathbf{A}=\mathbf{Q}$, $\mathbf{B}=\hat{\mathbf{G}}$, and $\mathbf{C}=\mathbf{P}$, respectively. Note that $\mathbf{v}_k$ can be further represented as $\mathbf{v}_k=[ \mathbf{v}_{1k}^{T},\dots ,\mathbf{v}_{Mk}^{T} ] ^T$ with $\mathbf{v}_{mk} \in \mathbb{C} ^{N} $. Relying on the block diagonal  characteristic of $\mathbf{B}_k$ and $\mathbf{C}_k$, we can also represent $\mathbf{Q}$ in the block diagonal form as $\mathbf{Q}=\mathrm{diag}[ \mathbf{Q}_1,\dots ,\mathbf{Q}_M ] $ with $\mathbf{Q}_m=\sum_{l=1}^K{p_l( \mathbf{B}_{ml}+\mathbf{B}_{ml}^{H}+\mathbf{C}_{ml} )}+\sigma ^2\mathbf{I}_N\in \mathbb{C} ^{N\times N}$. With the aid of the block diagonal  characteristic of $\mathbf{Q}$, we can further construct $\mathbf{v}_{mk}$ as  
\begin{equation}\label{vmk}
\mathbf{v}_{mk}=p_k\mathbf{Q}_{m}^{-1}\hat{\mathbf{G}}_m\left( \hat{\mathbf{G}}^H\mathbf{Q}^{-1}\hat{\mathbf{G}}+\mathbf{P}^{-1} \right) ^{-1}\mathbf{e}_k.
\end{equation}
where $\hat{\mathbf{G}}_m=\left[ \hat{\mathbf{g}}_{m1},\dots ,\hat{\mathbf{g}}_{mK} \right] \in \mathbb{C} ^{N\times K}$ denotes the local channel estimates at BS $m$.
Ideally, $\mathbf{v}_{mk}$ can be applied at each AP $m$ in a distributed manner. However, the direct implementation of \eqref{vmk} at each BS requires global instantaneous channel state information (CSI), which needs significant signaling transmission burden from other BSs via fronthaul links. To facilitate the practical 
implementation \eqref{vmk}, we apply the asymptotic analysis to approximate some terms in \eqref{vmk} by the statistical information and derive a GSLI-MMSE combining scheme (we call this scheme as GSLI-MMSE since it is obtained from the MMSE combining scheme) as in the following corollary.
\begin{coro}\label{GSLI}
An efficient distributed GSLI-MMSE combining scheme at BS $m$ for UE $k$ based on global statistics \& local instantaneous information is represented as
\begin{equation}\label{eqGSLI}
\mathbf{v}_{mk}=\frac{p_k}{MN}\mathbf{Q}_{m}^{-1}\hat{\mathbf{G}}_m\left( \mathbf{\Sigma }+\frac{1}{MN}\mathbf{P}^{-1} \right) ^{-1}\mathbf{e}_k,
\end{equation}
where the $(k,l)$-th element of $\mathbf{\Sigma }\in \mathbb{C} ^{K\times K}$ is
\begin{align}
\left[ \mathbf{\Sigma } \right] _{kl}\!=\!\begin{cases}
	\frac{1}{MN}\bar{\mathbf{g}}_{k}^{H}\mathbf{Q}^{-1}\bar{\mathbf{g}}_l, \,\,\, \mathrm{if} \,\,l\notin \mathcal{P} _k\,\,\\
	\frac{1}{MN}\bar{\mathbf{g}}_{k}^{H}\mathbf{Q}^{-1}\bar{\mathbf{g}}_l\!+\!\frac{1}{MN}\tau _p\mathrm{tr}(\mathbf{A}_l\mathbf{\Psi }_k\mathbf{A}_{k}^{H}\mathbf{Q}^{-1}). \, \mathrm{if} \,\, l\in \mathcal{P} _k\\
\end{cases}
\end{align}
\end{coro}
\begin{IEEEproof}
We first formulate \eqref{vmk} as 
\begin{equation}
\mathbf{v}_{mk}=\frac{p_k}{MN}\mathbf{Q}_{m}^{-1}\hat{\mathbf{G}}_m( \frac{1}{MN}\hat{\mathbf{G}}^H\mathbf{Q}^{-1}\hat{\mathbf{G}}+\frac{1}{MN}\mathbf{P}^{-1} ) ^{-1}\mathbf{e}_k.
\end{equation} 
Then, we focus on the term $\frac{1}{MN}\hat{\mathbf{G}}^H\mathbf{Q}^{-1}\hat{\mathbf{G}}\in \mathbb{C} ^{K\times K}$. Note that the $(k,l)$-th element of this term is
\begin{equation}\label{vitalterm}
\left[ \frac{1}{MN}\hat{\mathbf{G}}^H\mathbf{Q}^{-1}\hat{\mathbf{G}} \right] _{kl}=\frac{1}{MN}\hat{\mathbf{g}}_{k}^{H}\mathbf{Q}^{-1}\hat{\mathbf{g}}_l.
\end{equation}
For $l\notin \mathcal{P} _k$, we have
\begin{equation}\label{asymptoticreuslts1}
\begin{aligned}
&\frac{1}{MN}\hat{\mathbf{g}}_{k}^{H}\mathbf{Q}^{-1}\hat{\mathbf{g}}_l=\frac{1}{MN}( \bar{\mathbf{g}}_{k}^{H}+\hat{\mathbf{g}}_{k,\left( 1 \right)}^{H} ) \mathbf{Q}^{-1}( \bar{\mathbf{g}}_l+\hat{\mathbf{g}}_{l,\left( 1 \right)} ) \\
&\overset{\left( a \right)}{\asymp}\frac{1}{MN}\bar{\mathbf{g}}_{k}^{H}\mathbf{Q}^{-1}\bar{\mathbf{g}}_l,
\end{aligned}
\end{equation}
where step (a) follows from the standard asymptotic analysis result as \eqref{asy2} in Lemma~\ref{asymptotic} at the Appendix since $\hat{\mathbf{g}}_{k,\left( 1 \right)} $ and $\hat{\mathbf{g}}_{l,\left( 1 \right)} $ are mutually independent with $\hat{\mathbf{g}}_{k,\left( 1 \right)}=\mathbf{A}_k(\mathbf{y}_{k}^{p}-\bar{\mathbf{y}}_{k}^{p})$  and $\bar{\mathbf{g}}_{k}^{H}\mathbf{Q}^{-1}\hat{\mathbf{g}}_{l,\left( 1 \right)}\asymp 0$ \cite{sanguinetti2018theoretical}. For $l\in \mathcal{P} _k$, $\hat{\mathbf{g}}_{k,\left( 1 \right)}$ and $\hat{\mathbf{g}}_{l,\left( 1 \right)}$ are correlated. Following the standard result as \eqref{asy1} in Lemma~\ref{asymptotic} at the Appendix, we have 
\begin{equation}\label{asymptoticreuslts2}
\frac{1}{MN}\hat{\mathbf{g}}_{k}^{H}\mathbf{Q}^{-1}\hat{\mathbf{g}}_l\asymp \frac{1}{MN}[\bar{\mathbf{g}}_{k}^{H}\mathbf{Q}^{-1}\bar{\mathbf{g}}_l+\tau _p\mathrm{tr}\left( \mathbf{A}_l\mathbf{\Psi }_k\mathbf{A}_{k}^{H}\mathbf{Q}^{-1} \right)].
\end{equation}
Based on these derived results, we can easily derive the combining scheme as in \eqref{eqGSLI} by approximating $\frac{1}{MN}\hat{\mathbf{G}}^H\mathbf{Q}^{-1}\hat{\mathbf{G}}\in \mathbb{C} ^{K\times K}$ based on the asymptotic analysis-aided results as in \eqref{asymptoticreuslts1} and \eqref{asymptoticreuslts2}.
\end{IEEEproof}

\begin{rem}
As observed in Corollary~\ref{GSLI}, the combining scheme in \eqref{eqGSLI} is an asymptotic analysis-aided form of \eqref{vmk}, which can be computed based on the global statistics information, represented by $\mathbf{\Sigma }$, and local instantaneous information, represented by $\hat{\mathbf{G}}_m$. It is feasible to implement the combining scheme as in \eqref{eqGSLI} at each BS in a distributed manner, since only the global statistics, instead of global instantaneous, information is required in \eqref{eqGSLI}. Note that the channel statistics remain constant over a long period of time. Thus, to implement \eqref{eqGSLI} in a distributed manner, the global statistics information needs to be transmitted to each BS by only one time in each realization of the BS/UE locations. Then, each BS can locally design the GSLI-MMSE combining scheme based on Corollary~\ref{GSLI} with the aid of received global statistical information.
\end{rem}

\begin{rem}
Based on a similar idea to Corollary~\ref{GSLI}, a distributed combining scheme called reduced-complexity MMSE (RC-MMSE) combining is studied in \cite{polegre2021pilot} for the conventional CF mMIMO network. It is worth noting that the RC-MMSE combining in \cite{polegre2021pilot} only holds for the scenario with single-antenna BSs over the Rayleigh fading channel model, only with the NLoS component. Even if the MMSE channel estimator is applied, the GSLI-MMSE combining in this paper is still a more generalized version of the RC-MMSE combining, where the multiple-antenna BSs and the near-field spherical wave channel model with both the LoS and NLoS components are considered. These generalizations involve the matrix-level derivation methodology and more complicated channel characteristics, which require a different derivation method and are not a straightforward expansion of \cite{polegre2021pilot}. Note that the GSLI-MMSE combining in Corollary~\ref{GSLI} can reduce to the RC-MMSE combining as in \cite[Eq. (8)]{polegre2021pilot} by letting $\bar{\mathbf{g}}_{k}=\mathbf{0}$ and $N=1$.
\end{rem}

\subsection{SI-LMMSE Combining Scheme}
Note that we apply the asymptotic analysis method in Corollary~\ref{GSLI} to approximate the global instantaneous information-based terms. In this part, we propose a Statistics matrix Inversion-based LMMSE (SI-LMMSE) combining scheme. More specifically, we start from the LMMSE scheme as in \eqref{L-MMSE}. As observed in \eqref{L-MMSE}, only the term $\hat{\mathbf{g}}_{ml}\hat{\mathbf{g}}_{ml}^{H}$ in the matrix inversion part is based on the local instantaneous information, and the other terms are all based on the local statistics information. Notably, by approximating the term $\hat{\mathbf{g}}_{ml}\hat{\mathbf{g}}_{ml}^{H}$ by its long-term second-order expectation, we can derive the SI-LMMSE combining scheme as in the following corollary.

\begin{coro}\label{SL-MMSE}
By approximating $\hat{\mathbf{g}}_{ml}\hat{\mathbf{g}}_{ml}^{H}$ by its long-term second-order expectation, we derive the SI-LMMSE combining scheme as
\begin{equation}\label{eqSLMMSE}
\mathbf{v}_{mk}\!=\!p_k\!\left[ \sum_{l=1}^K{p_l\left( \bar{\mathbf{R}}_{ml}+\mathbf{B}_{ml}+\mathbf{B}_{ml}^{H}+\mathbf{C}_{ml} \right)}+\sigma ^2\mathbf{I}_N \right] ^{-1}\!\!\hat{\mathbf{g}}_{mk},
\end{equation}
where $\bar{\mathbf{R}}_{ml}\triangleq \mathbb{E} \left\{ \hat{\mathbf{g}}_{ml}\hat{\mathbf{g}}_{ml}^{H} \right\} =\bar{\mathbf{g}}_{ml}\bar{\mathbf{g}}_{ml}^{H}+\hat{\mathbf{R}}_{ml}$ is the expectation of $\hat{\mathbf{g}}_{ml}\hat{\mathbf{g}}_{ml}^{H}$.
\end{coro}

\begin{rem}
This combining scheme is mainly motivated from the single-cell MMSE (S-MMSE) combining scheme introduced in \cite[Eq. (4.8)]{8187178}, where some terms are replaced by their expectations. We can observe from Corollary~\ref{SL-MMSE} that the SI-LMMSE combining scheme can be designed based on the local channel information at each BS. More importantly, the SI-LMMSE is designed with the aid of the statistics-based matrix inversion, which embraces much lower computational complexity than that of the instantaneous information-based matrix inversion in the LMMSE combining scheme.
\end{rem}

\begin{rem}
Note that \eqref{eqSLMMSE} is a distributed combining scheme, which is derived from the LMMSE combining scheme as in \eqref{eq_LMMSE}. Notably, following the similar approach, a Statistics matrix inversion-based Centralized MMSE (SI-CMMSE) combining scheme can also be derived via the CMMSE as in \eqref{eq_CMMSE} by approximating $\hat{\mathbf{g}}_l\hat{\mathbf{g}}_{l}^{H}$ in \eqref{eq_CMMSE} as its expectation $\bar{\mathbf{g}}_l\bar{\mathbf{g}}_{l}^{H}+\hat{\mathbf{R}}_l\triangleq \bar{\mathbf{R}}_l$.
\end{rem}

\begin{algorithm}[t]
\label{Ins-SSOR}
\caption{Ins-SSOR-LMMSE combining scheme}
\KwIn{$\mathbf{A}_m^{\mathrm{Ins}}$ given in \eqref{A_Ins} and $\mathbf{b}_{mk}=p_k\hat{\mathbf{g}}_{mk}$ for all AP-UE pairs; the number of algorithm iterations $N_{Iter}$; relaxation parameter $\omega$ computed in \eqref{relaxation_para};}
\KwOut{Ins-SSOR-LMMSE combining scheme $\mathbf{v}_{mk}^{\mathrm{Ins}-\mathrm{SSOR}}$ for all AP-UE pairs;}
{\bf Initiation:} $i=0$, $\mathbf{v}_{mk}^{\mathrm{Ins}-\mathrm{SSOR},\left( 0 \right)}=\mathbf{0}$;\\

\Repeat(){ $i\geqslant N_{Iter}$}
{
$i=i+1$\\
Decompose $\mathbf{A}_m^{\mathrm{Ins}}$ based on \eqref{A_SSOR};\\
Update the first half iteration result $\mathbf{v}_{mk}^{\mathrm{Ins}-\mathrm{SSOR},( i+\frac{1}{2} )}$ via \eqref{SSOR_1};\\
Update the second half iteration results $\mathbf{v}_{mk}^{\mathrm{Ins}-\mathrm{SSOR},\left( i \right)}$ via \eqref{SSOR_2};\\
}
\end{algorithm}

\section{SSOR Aid Distributed Combining Design}\label{Sec_SSOR}
We notice that the LMMSE combining as in \eqref{L-MMSE} involves $N \times N$ instantaneous information-based matrix inversion with the computational complexity of $\mathcal{O} ( N^3 ) $ computation complexity, which is of high computational complexity when $N$ is large. To facilitate the practical implementation of this scheme, we apply a low-complexity SSOR algorithm to solve the computationally demanding matrix inversion.
\subsection{Fundamentals of the SSOR Algorithm}\label{SSOR_Fun}
In this part, we first introduce the basic fundamentals of the SSOR algorithm. We consider an SLE as $\mathbf{Ax}=\mathbf{b}$, where $\mathbf{A}\in \mathbb{C} ^{n_1\times n_1}$ is a Hermitian positive definite matrix and $\mathbf{x},\mathbf{b}\in \mathbb{C} ^{n_1}$. The fundamentals of the SSOR algorithm to solve this SLE are detailed as follows.

\begin{enumerate}

\item First, $\textbf{A}$ is decomposed as 
\begin{equation}\label{A_SSOR}
\mathbf{A}=\mathbf{D}+\mathbf{L}+\mathbf{L}^H,
\end{equation}
where $\mathbf{D}$ is the diagonal component of $\mathbf{A}$, $\mathbf{L}$ and $\mathbf{L}^H$ represent the strictly lower and upper triangular component of $\mathbf{A}$, respectively.

\item Then, we apply the successive over relaxation (SOR) method within the forward order to derive the first half iteration result as 
\begin{equation}\label{SSOR_1}
(\mathbf{D}+\omega \mathbf{L})\mathbf{x}^{(i+1/2)}=(1-\omega )\mathbf{Dx}^{(i)}-\omega \mathbf{L}^H\mathbf{x}^{(i)}+\omega \mathbf{b},
\end{equation}
where $i$ is the iteration number and $\omega$ is the relaxation parameter, which is computed as 
\begin{equation}\label{relaxation_para}
\omega =\frac{2}{1+\sqrt{2\left( 1-\mu \right)}}
\end{equation}
with $\mu =( 1+\sqrt{{K}/{N}} ) ^2-1$ \cite{xie2016low}.

\item Finally, we apply the SOR method within the reverse order to obtain the second half iteration result as 
\begin{equation}\label{SSOR_2}
(\mathbf{D}+\omega \mathbf{L}^H)\mathbf{x}^{(i+1)}=(1-\omega )\mathbf{Dx}^{(i+1/2)}-\omega \mathbf{Lx}^{(i+1/2)}+\omega \mathbf{b}.
\end{equation}
\end{enumerate}

\begin{rem}
By applying this SSOR method, the highly computational demanding matrix inversion with $\mathcal{O} ( N^3 ) $ computational complexity can be avoided since only the total computational complexity of $\mathcal{O} ( N^2 ) $ \cite{xie2016low} is required for each iteration in \eqref{SSOR_1} and \eqref{SSOR_2}.
\end{rem}

\subsection{Ins-SSOR-LMMSE Combining Scheme}\label{sec_ins_ssor}
In this part, we introduce an Instantaneous SSOR algorithm aided LMMSE (Ins-SSOR-LMMSE) combining scheme $\mathbf{v}_{mk}^{\mathrm{Ins}-\mathrm{SSOR}}$. More specifically, the SSOR algorithm introduced in Sec.~\ref{SSOR_Fun} is applied to solve the instantaneous information-based matrix inversion in \eqref{eq_LMMSE}. We let the variables of the SSOR algorithm in Sec.~\ref{SSOR_Fun} be 
\begin{equation}\label{A_Ins}
\mathbf{A}= \sum_{l=1}^K{p_l( \hat{\mathbf{g}}_{ml}\hat{\mathbf{g}}_{ml}^{H}+\mathbf{B}_{ml}+\mathbf{B}_{ml}^{H}+\mathbf{C}_{ml} )}+\sigma ^2\mathbf{I}_N\triangleq \mathbf{A}_m^{\mathrm{Ins}},
\end{equation} 
$\mathbf{b}=p_k\hat{\mathbf{g}}_{mk}\triangleq\mathbf{b}_{mk}$, and $\mathbf{x}=\mathbf{v}_{mk}^{\mathrm{Ins}-\mathrm{SSOR}}$ respectively. Thus, following the steps of the SSOR algorithm, the main steps of the Ins-SSOR-LMMSE combining scheme $\mathbf{v}_{mk}^{\mathrm{Ins}-\mathrm{SSOR}}$ are summarized in Algorithm~\ref{Ins-SSOR}.

\begin{rem}
As observed from Algorithm~\ref{Ins-SSOR}, all variables, including $\mathbf{A}_m^{\mathrm{Ins}}$, are based on the instantaneous information. Thus, the matrix inversion of $(\mathbf{D}^{\mathrm{Ins}}+\omega \mathbf{L}^{\mathrm{Ins}})$ or $(\mathbf{D}^{\mathrm{Ins}}+\omega \mathbf{L}^{\mathrm{Ins},H})$ with $\mathbf{A}^{\mathrm{Ins}}=\mathbf{D}^{\mathrm{Ins}}+\mathbf{L}^{\mathrm{Ins}}+\mathbf{L}^{\mathrm{Ins},H}$ should be computed for every channel realization.
\end{rem}

\vspace{-0.3cm}

\begin{algorithm}[t]
\label{Sta-SSOR}
\caption{Sta-SSOR-LMMSE combining scheme}
\KwIn{$\mathbf{A}_m^{\mathrm{Sta}}$ given in \eqref{A_Sta} and $\mathbf{b}_{mk}=p_k\hat{\mathbf{g}}_{mk}$ for all AP-UE pairs; the number of algorithm iterations $N_{Iter}$; relaxation parameter $\omega$ computed in \eqref{relaxation_para};}
\KwOut{Ins-SSOR-LMMSE combining scheme $\mathbf{v}_{mk}^{\mathrm{Ins}-\mathrm{SSOR}}$ for all AP-UE pairs;}
{\bf Initiation:} $i=0$, $\mathbf{v}_{mk}^{\mathrm{Ins}-\mathrm{SSOR},\left( 0 \right)}=\mathbf{0}$;\\

\Repeat(){ $i\geqslant N_{Iter}$}
{
$i=i+1$\\
Decompose $\mathbf{A}_m^{\mathrm{Sta}}$ based on \eqref{A_SSOR};\\
Update the first half iteration result $\mathbf{v}_{mk}^{\mathrm{Sta}-\mathrm{SSOR},( i+\frac{1}{2} )}$ via \eqref{SSOR_1};\\
Update the second half iteration results $\mathbf{v}_{mk}^{\mathrm{Sta}-\mathrm{SSOR},\left( i \right)}$ via \eqref{SSOR_2};\\
}
\end{algorithm}

\begin{table*}[t!]
  \centering
  \fontsize{8.5}{11}\selectfont
  \caption{Comparison of all studied combining schemes in this paper.}
  \vspace*{-0.3cm}
  \label{Comparisons}
   \begin{tabular}{ !{\vrule width0.7 pt}  m{1.65 cm}<{\centering} !{\vrule width0.7pt} m{3.7 cm}<{\centering} !{\vrule width0.7pt} m{2.7 cm}<{\centering} !{\vrule width0.7pt}  m{3.5 cm}<{\centering} !{\vrule width0.7pt}  m{3.4cm}<{\centering} !{\vrule width0.7pt}}

    \Xhline{0.7pt}
         \bf Scheme &\bf Basic idea &\bf Applied information & \bf Computational complexity of combining vector computation  & \bf Computational complexity of precomputation based on statistics \cr
    \Xhline{0.7pt}
    CMMSE &Minimizing the \textbf{global conditional MSE} &Global instantaneous information & $\mathcal{O} ( M^2N^2KN_r+M^3N^3N_r ) $ & $\mathcal{O} ( MN^3K )$\cr\hline
    GSLI-MMSE &Applying the \textbf{asymptotic analysis} to approximate the \textbf{global instantaneous} information matrices &Global statistics \& local instantaneous information  & $\mathcal{O} ( MN^3+MN^2KN_r+K^3+MNK^2N_r ) $   & $\mathcal{O} ( M^3N^3K^2 ) $\cr\hline
    LMMSE &Minimizing the \textbf{local conditional MSE} &Local instantaneous information & $\mathcal{O} ( MN^2KN_r+MN^3N_r ) $   & $\mathcal{O} ( MN^3K )$\cr\hline
    SI-LMMSE &Applying the \textbf{matrix expectations} to approximate \textbf{local instantaneous} information matrices  &Local instantaneous information \& local statistics information matrix inversion & $\mathcal{O} ( MN^3+MN^2KN_r ) $ &$\mathcal{O} ( MN^3K )$\cr\hline
    Ins-SSOR &Applying the \textbf{SSOR algorithm} to solve the \textbf{local instantaneous} information matrix inversion  &Local instantaneous information &$\mathcal{O} ( MN^2KN_rN_{Iter} ) $ & $\mathcal{O} ( MN^3K )$  \cr\hline
    Sta-SSOR &Applying the \textbf{SSOR algorithm} to solve the \textbf{local statistics} information matrix inversion &Local instantaneous information \& local statistics information matrix inversion &$\mathcal{O} ( MN^2KN_rN_{Iter} ) $ & $\mathcal{O} ( MN^3K )$  \cr\hline
    Ins-SI-SSOR &Involving the \textbf{statistics information-based initial value} to the Ins-SSOR scheme &Local instantaneous information \& local statistics information-based initial value &$\mathcal{O} ( MN^2KN_rN_{Iter}+MN^3 ) $ & $\mathcal{O} ( MN^3K )$  \cr\hline

    \Xhline{0.7pt}
    \end{tabular}
  \vspace*{-0.6cm}
\end{table*}

\subsection{Sta-SSOR-LMMSE Combining Scheme}\label{sec_sta_ssor}
Then, we investigate a Statistics matrix inversion SSOR algorithm aided LMMSE (Sta-SSOR-LMMSE) combining scheme $\mathbf{v}_{mk}^{\mathrm{Sta}-\mathrm{SSOR}}$. More specifically, we apply the SSOR algorithm to solve the statistics matrix inversion in the SI-LMMSE combining scheme as in 
\eqref{eqSLMMSE} by letting variables of the SSOR algorithm in Sec.~\ref{SSOR_Fun} as  
\begin{equation}\label{A_Sta}
\mathbf{A}= \sum_{l=1}^K{p_l( \bar{\mathbf{R}}_{ml}+\mathbf{B}_{ml}+\mathbf{B}_{ml}^{H}+\mathbf{C}_{ml} )}+\sigma ^2\mathbf{I}_N\triangleq \mathbf{A}_m^{\mathrm{Sta}},
\end{equation} 
$\mathbf{b}=p_k\hat{\mathbf{g}}_{mk}\triangleq\mathbf{b}_{mk}$, and $\mathbf{x}=\mathbf{v}_{mk}^{\mathrm{Sta}-\mathrm{SSOR}}$ respectively. Based on the fundamentals of the SSOR algorithm, we summarize the main steps of the Sta-SSOR-LMMSE combining scheme $\mathbf{v}_{mk}^{\mathrm{Sta}-\mathrm{SSOR}}$ in Algorithm~\ref{Sta-SSOR}.

\begin{rem}
Compared with the Ins-SSOR-LMMSE combining as in Algorithm~\ref{Ins-SSOR}, where $\mathbf{A}_m^{\mathrm{Ins}}$ is based on the instantaneous information, $\mathbf{A}_m^{\mathrm{Sta}}$ in Algorithm~\ref{Sta-SSOR} is only based on the statistics information, which remains constant over in each realization of the BS/UE locations. Thus, the matrix inversion of the $(\mathbf{D}^{\mathrm{Sta}}+\omega \mathbf{L}^{\mathrm{Sta}})$ or $(\mathbf{D}^{\mathrm{Sta}}+\omega \mathbf{L}^{\mathrm{Sta},H})$ with $\mathbf{A}^{\mathrm{Sta}}=\mathbf{D}^{\mathrm{Sta}}+\mathbf{L}^{\mathrm{Sta}}+\mathbf{L}^{\mathrm{Sta},H}$ only needs to be computed once for each realization of the BS/UE locations, which involves much smaller computational complexity than that of the Ins-SSOR-LMMSE combining scheme.
\end{rem}

% \begin{table}[t!]
%   \centering
%   \fontsize{8.5}{11}\selectfont
%   \caption{Computational complexity for each realization of the BS/UE locations for all studied combining schemes. }
%   \vspace*{-0.3cm}
%   \label{Comparisons}
%    \begin{tabular}{ !{\vrule width0.7 pt}  m{1.7 cm}<{\centering} !{\vrule width0.7pt}  m{3.3 cm}<{\centering} !{\vrule width0.7pt}  m{2.6cm}<{\centering} !{\vrule width0.7pt}}

%     \Xhline{0.7pt}
%          \bf Scheme & \bf Combining vector computation  & \bf Precomputation based on statistics \cr
%     \Xhline{0.7pt}
%     CMMSE & $\mathcal{O} ( M^2N^2KN_r+M^3N^3N_r ) $ & $\mathcal{O} ( MN^3K )$\cr\hline
%     GSLI-MMSE & $\mathcal{O} ( MN^3+MN^2KN_r+K^3+MNK^2N_r ) $   & $\mathcal{O} ( M^3N^3K^2 ) $\cr\hline
%     LMMSE & $\mathcal{O} ( MN^2KN_r+MN^3N_r ) $   & $\mathcal{O} ( MN^3K )$\cr\hline
%     SI-LMMSE & $\mathcal{O} ( MN^3+MN^2KN_r ) $ &$\mathcal{O} ( MN^3K )$\cr\hline
%     Ins-SSOR &$\mathcal{O} ( MN^2KN_rN_{Iter} ) $ & $\mathcal{O} ( MN^3K )$  \cr\hline
%     Sta-SSOR &$\mathcal{O} ( MN^2KN_rN_{Iter} ) $ & $\mathcal{O} ( MN^3K )$  \cr\hline
%     Ins-SI-SSOR &$\mathcal{O} ( MN^2KN_rN_{Iter}+MN^3 ) $ & $\mathcal{O} ( MN^3K )$  \cr\hline

%     \Xhline{0.7pt}
%     \end{tabular}
%   \vspace*{-0.6cm}
% \end{table}

\begin{algorithm}[t]
\label{Ins-SI-SSOR}
\caption{Ins-SI-SSOR-LMMSE combining scheme}
\KwIn{$\mathbf{A}_m^{\mathrm{Ins}}$ given in \eqref{A_Ins} and $\mathbf{b}_{mk}=p_k\hat{\mathbf{g}}_{mk}$ for all AP-UE pairs; the number of algorithm iterations $N_{Iter}$; relaxation parameter $\omega$ computed in \eqref{relaxation_para};}
\KwOut{Ins-SI-SSOR-LMMSE combining scheme $\mathbf{v}_{mk}^{\mathrm{Ins}-\mathrm{SI}-\mathrm{SSOR}}$ for all AP-UE pairs;}
{\bf Initiation:} $i=0$, $\mathbf{v}_{mk}^{\mathrm{Ins}-\mathrm{SI}-\mathrm{SSOR},\left( 0 \right)}=p_k\mathbf{A}_{m}^{\mathrm{Sta},-1}\hat{\mathbf{g}}_{mk}$;\\

\Repeat(){ $i\geqslant N_{Iter}$}
{
$i=i+1$\\
Decompose $\mathbf{A}_m^{\mathrm{Ins}}$ based on \eqref{A_SSOR};\\
Update the first half iteration result $\mathbf{v}_{mk}^{\mathrm{Ins}-\mathrm{SI}-\mathrm{SSOR},( i+\frac{1}{2} )}$ via \eqref{SSOR_1};\\
Update the second half iteration results $\mathbf{v}_{mk}^{\mathrm{Ins}-\mathrm{SI}-\mathrm{SSOR},\left( i \right)}$ via \eqref{SSOR_2};\\
}
\end{algorithm}
\subsection{Ins-SI-SSOR-LMMSE Combining Scheme}
As observed in above two SSOR method aided combining schemes, the initial values $\mathbf{v}_{mk}^{\mathrm{Ins}-\mathrm{SSOR},\left( 0 \right)}$, $\mathbf{v}_{mk}^{\mathrm{Sta}-\mathrm{SSOR},\left( 0 \right)}$ are set as $\mathbf{0}$. To enhance the achievable performance and accelerate the convergence speed of the SSOR algorithm, motivated by \cite{liu2024fast}, we propose an Instantaneous SSOR algorithm with Statistics-based Initial value aided LMMSE (Ins-SI-SSOR-LMMSE) combining scheme $\mathbf{v}_{mk}^{\mathrm{Ins}-\mathrm{SI}-\mathrm{SSOR}}$. More specifically, we define the initial value of the SSOR algorithm as $\mathbf{x}^{\left( 0 \right)}=\mathbf{v}_{mk}^{\mathrm{Ins}-\mathrm{SI}-\mathrm{SSOR},\left( 0 \right)}=p_k\mathbf{A}_{m}^{\mathrm{Sta},-1}\hat{\mathbf{g}}_{mk}$ with $\mathbf{A}_{m}^{\mathrm{Sta}}$ given as in \eqref{A_Sta}. Then, we apply the instantaneous information-based matrix $\mathbf{A}_m^{\mathrm{Ins}}$ to design the Ins-SI-SSOR-LMMSE combining scheme, and the main steps are provided in Algorithm~\ref{Ins-SI-SSOR}.

\begin{rem}
As examined in the following simulation results, the introduction of statistics matrix-based initial value can not only improve the performance but also accelerate the convergence speed of the SSOR method. Besides, the involvement of the initial value requires $N \times N$  matrix inversion with the computational complexity of $\mathcal{O} ( N^3 ) $ computation complexity. However, the matrix inversion in the initial value computation is based on the statistics information and thus is only needed to be computed once for each realization of the BS/UE positions, which is of much smaller computational complexity than the LMMSE combining in \eqref{eq_LMMSE}.
\end{rem}

\begin{rem}
This part provides a potential initial value choice for the SSOR algorithm. Note that the zero initial value choices utilized in Sec.~\ref{sec_ins_ssor} and Sec.~\ref{sec_sta_ssor} are simple and robust but not optimal choices. Balancing the achievable performance and coverage speed, as discussed in the simulation results, we advocate for utilizing the statistical information-based initial values for the instantaneous SSOR algorithm-aided LMMSE schemes. As for the Sta-SSOR-LMMSE combining scheme, it would be insightful to explore some more potential initial value choices in the future. For instance, in scenarios with low user mobility, we can apply the previously computed combining scheme results as the initial values in a new round of SSOR iteration to further involve the iteration savings and performance enhancement than the zero initial value scenario.
\end{rem}

\section{Overview and Practical Implementation Guidance for Studied Combining Schemes}

To have a comprehensive overview of all studied combining schemes in this paper, we compare all studied combining schemes from the perspective of the basic idea, applied information type, and computational complexity in Table~\ref{Comparisons}.  Note that the terminologies Ins-SSOR, Sta-SSOR, and Ins-SI-SSOR denote respective SSOR-based LMMSE combining schemes introduced above. The computational complexity is computed for each realization of the BS/UE locations, where $N_{Iter}$ is the iteration number of the SSOR algorithm. We assume that the statistical information varies with each realization of the BS/UE locations, and $N_r$ is the number of instantaneous channel realizations per realization of the BS/UE locations. Moreover, unless otherwise stated, the matrix inversion in one particular scheme is implemented based on the information type in the ``Applied information", where $I_r$ denotes the number of channel realizations for each realization of the AP/UE locations.

To provide practical implementation guidance on when each proposed combining scheme is advocated to be used, we relate them to the available information, the per coherence block computational/latency budget, and the stability of second-order statistics.  The GSLI-MMSE is recommended when the approaching CMMSE performance is desired, but exchanging global instantaneous CSI is infeasible, while exchanging slowly varying global statistics is acceptable. 
The SI-LMMSE is recommended under stringent online complexity/latency constraints, since the statistics-based terms can be reused across many instantaneous channel realizations. However, it cannot effectively exploit instantaneous channel realization-dependent interference suppression and may suffer a larger gap to the LMMSE scheme with a large number of array antennas, as will be demonstrated in Fig.~\ref{FIG_5_SI_K}. 
The SSOR-based schemes provide a tunable tradeoff between performance and complexity through the number of iterations. The Ins-SSOR-LMMSE is suitable when per-block iterative processing is affordable, and one aims to approach the LMMSE without explicit matrix inversion. The Sta-SSOR-LMMSE is preferable when second-order statistics are stable and further online savings are required by reusing a statistics-based system matrix, at the cost of additional mismatch under rapidly varying instantaneous conditions. The Ins-SI-SSOR-LMMSE is recommended when the iteration budget is small, but one still targets excellent performance, since the statistical information-based warm start accelerates convergence and yields a more favorable accuracy-complexity tradeoff with few iterations.

\begin{rem}\label{general}
The proposed combining schemes are not specifically tailored to any particular near-field channel model. Instead, they are formulated in terms of the estimated channel response $\hat{\mathbf{g}}_{mk}$ and long-term second-order statistics, such as the channel covariance matrices-related terms. Therefore, when these inputs are obtained under the near-field spherical-wave model and the mutual coupling as described in Section~\ref{system}, the corresponding near-field and mutual coupling effects are embedded in $\hat{\mathbf{g}}_{mk}$ and the associated statistics, and thus, the resulting combining vectors naturally incorporate such characteristics. Note that all studied combining schemes with tractable expressions hold not only for the near-field channel model, but also for the generalized channel models distributed as $\mathbf{g}_{mk}\sim \mathcal{N} _{\mathbb{C}}( \bar{\mathbf{g}}_{mk},\check{\mathbf{R}}_{mk} ) $. For instance, $\check{\mathbf{R}}_{mk}$ can be modeled based on the Weichselberger model, the Kronecker model, and so on \cite{04962}. This fact showcases the generalizability and scalability of the studied combining schemes.
\end{rem}

\begin{rem}
The proposed schemes can be implemented in a distributed manner, where each BS can perform local combining with its own instantaneous CSI and different levels of statistical information. These schemes are robust to limited channel bandwidth and non-negligible latency, since only long-term statistical information needs to be exchanged between different BSs. However, in networks with limited resources (e.g., limited channel bandwidth or severe latency), spatial division multiplexing (SDM) and over-the-air (OTA) techniques can be flexibly utilized to enhance the practical feasibility of the studied combining schemes.
\end{rem}

\section{Numerical Results}\label{numerical}
We consider a CF XL-MIMO system, where BSs and UEs are randomly located at a $1\times 1 \, \mathrm{km}^2$ area. All BSs and UEs are assumed to have LoS paths, and thus the large-scale fading coefficient for the LoS and NLoS parts can be modeled as $\beta _{mk}^{\mathrm{LoS}}=\frac{\kappa _{mk}}{\kappa _{mk}+1}\beta _{mk}$ and $\beta _{mk}^{\mathrm{NLoS}}=\frac{1}{\kappa _{mk}+1}\beta _{mk}$, respectively, where $\beta _{mk}$ is modeled based on the 3GPP COST 231 Walfish-Ikegami model as in \cite[Eq. (5.2-3)]{3GPP2024} and $\kappa _{mk}=10^{1.3-0.003d_{mk}}$ is the Rician $\kappa$ factor between AP $m$ and UE $k$. Moreover, we consider the BS and UE heights as $L_{BS}=12.5 \, \mathrm{m}$ and $L_{UE}=1.5 \, \mathrm{m}$, respectively. Besides, we have the carrier frequency $f_c=3 \, \mathrm{GHz}$, $\tau _c=200
$, $\tau _p=1$, $p_k=200 \, \mathrm{mW}$, and $\sigma ^2=-94 \, \mathrm{dBm}$, respectively. For the parameters of the mutual coupling effect, we have $\eta =120\,\pi $, $\gamma _0=0.577$, $Z_L=50\,\Omega$, $\Delta_l=0.1\lambda $, and, $r=10^{-5}\lambda $, respectively. Unless mentioned, we have $N_{Iter}=5$, and we apply the MMSE channel estimator. For each realization of the BS/UE locations, there are $N_r=800$ channel realizations. Besides, we consider the square UPA at each BS, where $N_x=N_y$ and $\Delta_x=\Delta_y$. Note that the network scale can be characterized by the system load metric---the ratio between the total BS antennas and the total UE antennas, that is $MN_xN_y/K$ in this paper. The scenarios with relatively large or small load ratios can showcase the scenarios with extremely large antenna arrays and with ultra-high user density, respectively. Based on this criteria, one can access to various network scales in the following simulation results. Meanwhile, different with the single-site XL-MIMO network, where hundreds or even thousands of antennas are provisioned at single BS, the CF XL-MIMO network distributes the antenna array across the cooperated BSs within the coverage area. Thus, under this cooperative paradigm, the effective antenna number in CF XL-MIMO is the total number of BS antennas, that is $MN_xN_y$. In the following, our simulations will involve CF XL-MIMO configurations ranging from hundreds to more than two thousand antennas.

In Fig.~\ref{FIG_1_GSLI_N}, we investigate the average SE performance for the GSLI-MMSE combining scheme against $N_x$. As observed, the GSLI-MMSE combining scheme showcases excellent SE performance approaching that of the CMMSE combining scheme and significantly outperforms the LMMSE combining scheme. For instance, with $N_x=16$, there is only $0.38 \%$ performance gap between the GSLI-MMSE and CMMSE combining schemes and $42.38 \%$ performance improvement between the GSLI-MMSE and LMMSE combining schemes, respectively. Moreover, the SE performance gaps between the GSLI-MMSE and CMMSE combining schemes becomes smaller as $N_x$ increases, where there are $1.96 \%$ and $0.42 \%$ performance gaps for $N_x=4$ and $N_x=16$, respectively. This observation showcases that more excellent performance can be achieved by the GSLI-MMSE combining when the number of antennas is large, which demonstrates the asymptotic behavior in Corollary~\ref{GSLI}. Notably, the GSLI-MMSE also performs well in the scenario with a small number of antennas, which further demonstrates its promising applicability. All these findings demonstrate that the proposed GSLI-MMSE combining scheme can showcase excellent performance with only additional global statistics information compared to the LMMSE combining scheme.

Fig.~\ref{FIG_2_GSLI_M} shows the sum SE performance for the GSLI-MMSE combining scheme against $M$. We observe that the GSLI-MMSE showcases the excellent SE performance in both scenarios with small and large $M$. Moreover, the performance gap between the GSLI-MMSE and LMMSE combining schemes becomes larger as $M$ increases. For instance, $25.95 \%$ and $44.63 \%$ SE improvement can be achieved by the GSLI-MMSE combining compared to the LMMSE combining for $M=2$ and $M=6$, respectively. This is because the GSLI-MMSE combining scheme can utilize the global statistics information from all BSs to facilitate the SE performance. Meanwhile, the performance gaps between the GSLI-MMSE and the C-MMSE combining schemes are tiny for different numbers of BSs.

In Fig.~\ref{FIG_3_Estimator}, we study the average SE performance for the GSLI-MMSE combining scheme under different channel estimators. The UatF capacity bound is applied to evaluate the achievable SE performance for the centralized processing since the standard capacity bound as in \eqref{SINR_centrlized_Standard} only holds for the MMSE channel estimators. EW-MMSE and generalized least-square (GLS) estimators are involved with $\mathbf{A}_{mk}=\sqrt{p_k}\mathbf{D}_{mk}\mathbf{\Gamma }_{mk}^{-1}$ and $\mathbf{A}_{mk}=1/( \sqrt{p_k}\tau _p ) \mathbf{I}_N$, respectively, with $\mathbf{D}_{mk}=\mathrm{diag}( [ \check{\mathbf{R}}_{mk} ] _{nn}: n=1,\dots ,N )$ and $\mathbf{\Gamma }_{mk}=\mathrm{diag}( [ \mathbf{\Psi }_{mk} ] _{nn}: n=1,\dots ,N )$. As observed, the MMSE estimator can achieve the highest SE among these estimators for any choice of the combining scheme. Besides, the GSLI-MMSE combining performs well under all considered channel estimators and the performance gap between the GSLI-MMSE and CMMSE combining is smallest under the MMSE estimator, that is only 
$0.59 \%$.

\begin{figure}[t]
\centering
\includegraphics[scale=0.5]{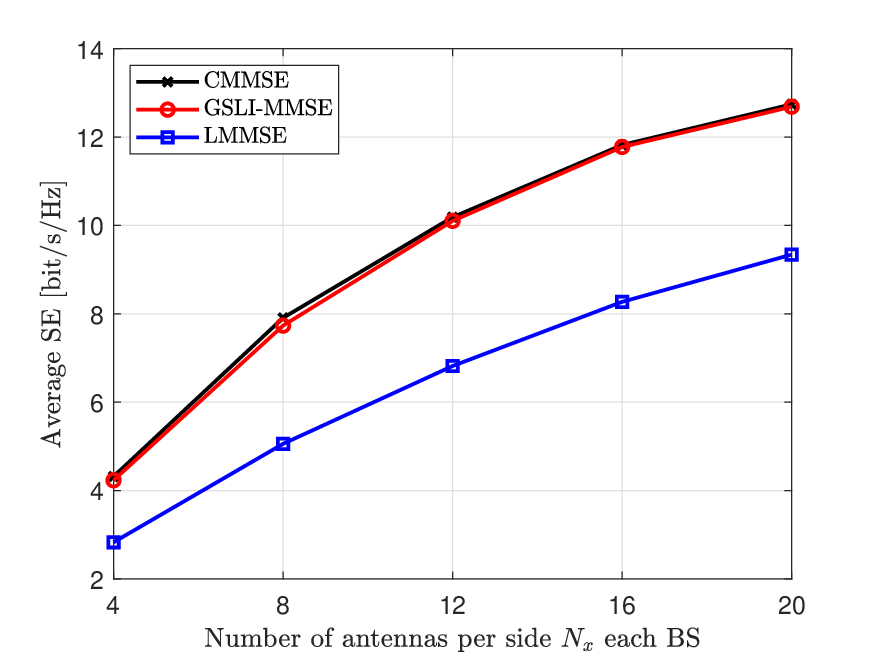}
\caption{Average SE performance for the GSLI-MMSE combining scheme against the number of antennas per side each BS $N_x$ with $M=4$, $K=20$, and $\Delta_x=\lambda/4$. \label{FIG_1_GSLI_N}}
\vspace{-0.3cm}
\end{figure}

\begin{figure}[t]
\centering
\includegraphics[scale=0.5]{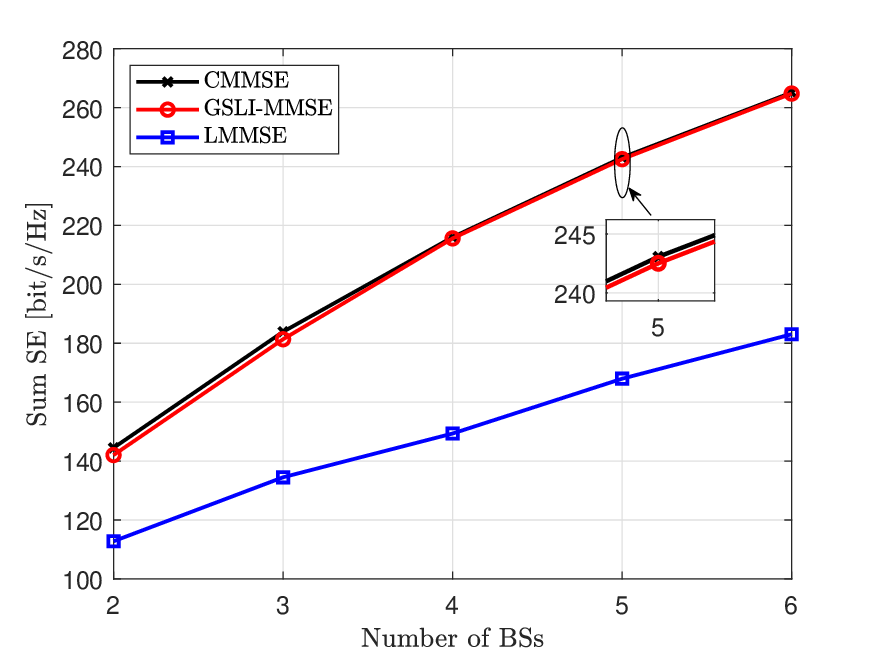}
\caption{Sum SE performance for the GSLI-MMSE combining scheme against the number of BSs $M$ with $N_x=16$, $K=20$, and $\Delta_x=\lambda/4$. \label{FIG_2_GSLI_M}}
\vspace{-0.3cm}
\end{figure}

\begin{figure}[t]
\centering
\includegraphics[scale=0.5]{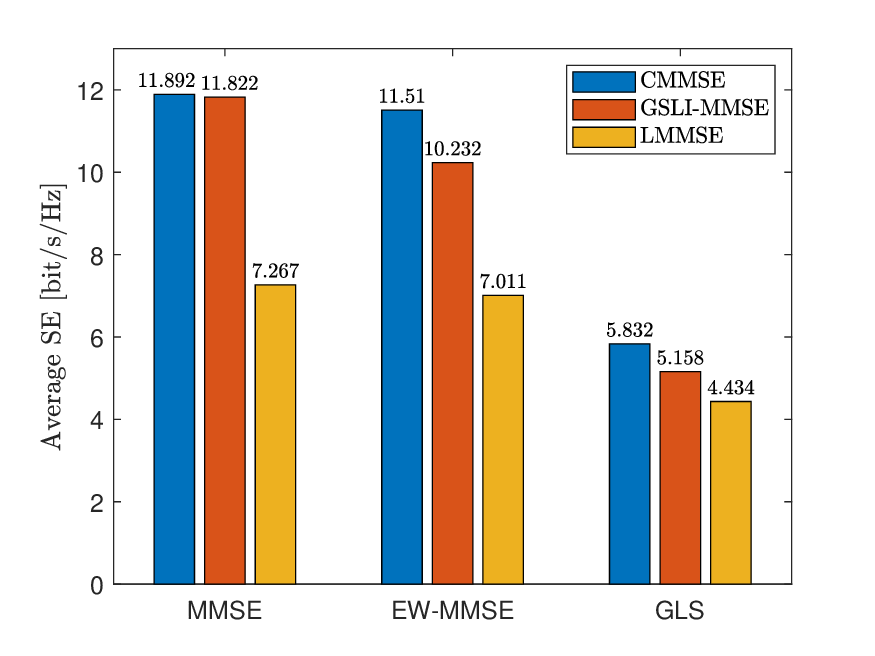}
\caption{Average SE performance evaluated by the UatF bound for the GSLI-MMSE combining scheme under different channel estimators with $M=8$, $N_x=8$, $K=20$, and $\Delta_x=\lambda/4$. \label{FIG_3_Estimator}}
\vspace{-0.3cm}
\end{figure}

\begin{figure}[t]
\centering
\includegraphics[scale=0.5]{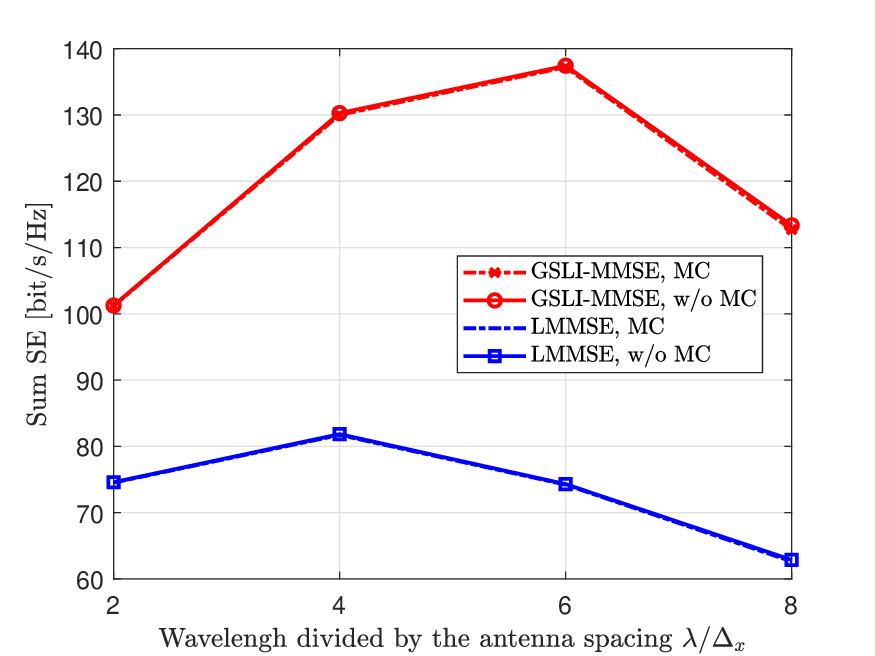}
\caption{Sum SE performance for the GSLI-MMSE combining scheme against $\lambda /\Delta _x$ with $M=8$, $N_x=8$, and $K=20$. \label{FIG_4_MC}}
\vspace{-0.3cm}
\end{figure}

\begin{figure}[t]
\centering
\includegraphics[scale=0.5]{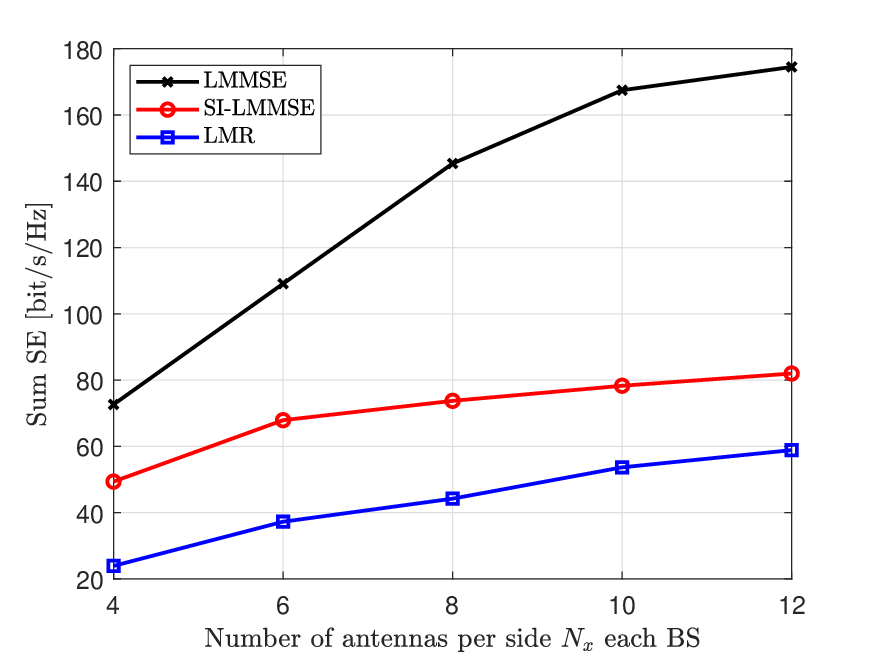}
\caption{Sum SE performance for the SI-LMMSE combining scheme against the number of antennas per side each BS $N_x$ with $M=6$ $K=20$, and $\Delta_x=\lambda/4$. \label{FIG_5_SI_K}}
\vspace{-0.3cm}
\end{figure}

Fig.~\ref{FIG_4_MC} shows the sum SE performance for the GSLI-MMSE combining scheme against $\lambda /\Delta _x$, which demonstrates that the larger value of $\lambda /\Delta _x$ denotes the smaller antenna spacing. As observed in Fig.~\ref{FIG_4_MC}, there is an optimal antenna spacing, which can achieve the best SE performance. For instance, for the GSLI-MMSE combining scheme, the performance gap between $\lambda /\Delta _x=2$ and  $\lambda /\Delta _x=6$ is $35.11 \%$ in the scenario without the mutual coupling effect. This observation demonstrates that, to facilitate the performance, the antennas at each BS should be deployed in a manner with proper antenna spacing. However, we confirm that this observation in this paper is directly derived from the simulation results, which is not further studied since it is not the key technical issue of this paper. This phenomenon is also observed in \cite[Fig. 12]{liu2023double}, which is advocated to be further studied in the future. Moreover, we find that the mutual coupling effect may not significantly degrade the achievable performance. Even for $\Delta _x=\lambda/8$, there is only $0.24 \%$ performance gap between the scenarios without and with the mutual coupling effect for the GSLI-MMSE combining.

Fig.~\ref{FIG_5_SI_K} investigates the sum SE performance for the SI-LMMSE combining scheme against $N_x$. As observed, the SI-LMMSE performs well, especially when the number of antennas each BS is small. For $N_x=4$ and $N_x=12$, the performance gaps between the SI-LMMSE and LMMSE combining schemes are $32.04 \%$ and $53.02 \%$, respectively. Meanwhile, the performance improvements between the SI-LMMSE and LMR combining schemes are $106.35 \%$ and $39.26 \%$ for $N_x=4$ and $N_x=12$, respectively. These observations stem from the fact that the approximation in SI-MMSE as \eqref{eqSLMMSE} removes the instantaneous interference-subspace information carried by $\hat{\mathbf{g}}_{ml}\hat{\mathbf{g}}_{ml}^{H}$. When the array dimension grows, the LMMSE combining in \eqref{eq_LMMSE} can increasingly exploit the instantaneous channel directions to suppress interference more effectively, whereas the SI-LMMSE combining relies on the average subspace described by $\bar{\mathbf{R}}_{ml}$ and thus cannot effectively exploit the instantaneous variability of the interference subspace that is encoded in $\hat{\mathbf{g}}_{ml}\hat{\mathbf{g}}_{ml}^{H}$. However, the SI-LMMSE combining scheme showcases an excellent trade-off between the achievable performance and the computational complexity. The computational complexity of this scheme is significantly smaller than that of the LMMSE combining scheme (this argument can be observed from Table~\ref{Comparisons} and will also be verified in the following) and the SI-LMMSE scheme significantly outperforms the LMR scheme.

Next, we focus on three SSOR algorithm-based combining schemes. In Fig.~\ref{FIG_6_SSOR_Pillar}, we study the average SE performance for the SSOR algorithm-based combining schemes with different $N_x$. As observed, the Ins-SI-SSOR combining scheme (we omit the terminology ``LMMSE" for simplicity) can achieve the best SE performance among three SSOR-based combining schemes and there is only $12.68 \%$ performance gap between the Ins-SI-SSOR and LMMSE combining schemes for $N_x=8$. Meanwhile, the Ins-SSOR and Sta-SSOR combining schemes showcase $81.16 \%$ and $23.67 \%$ SE improvement compared with the local maximum ratio (LMR) combining scheme $\mathbf{v}_{mk}=\hat{\mathbf{g}}_{mk}$, respectively, for $N_x=16$. Besides, the Ins-SSOR combining scheme outperforms the Sta-SSOR combining scheme since the Sta-SSOR combining scheme is a matrix expectation-aided approximation version of the Ins-SSOR combining scheme. And we observe that the performance gap between the Sta-SSOR and Ins-SSOR combining schemes becomes larger as $N_x$ increases.

\begin{figure}[t]
\centering
\includegraphics[scale=0.5]{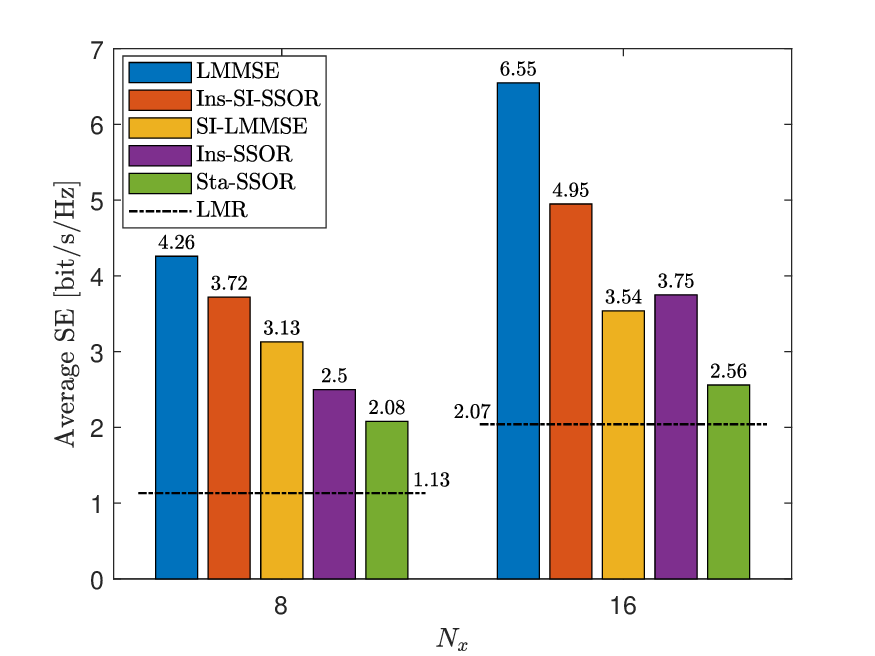}
\caption{Average SE performance for the SSOR algorithm-based combining schemes with $M=8$, $K=10$, and $\Delta_x=\lambda/8$. \label{FIG_6_SSOR_Pillar}}
\vspace{-0.3cm}
\end{figure}

\begin{figure}[t]
\centering
\includegraphics[scale=0.5]{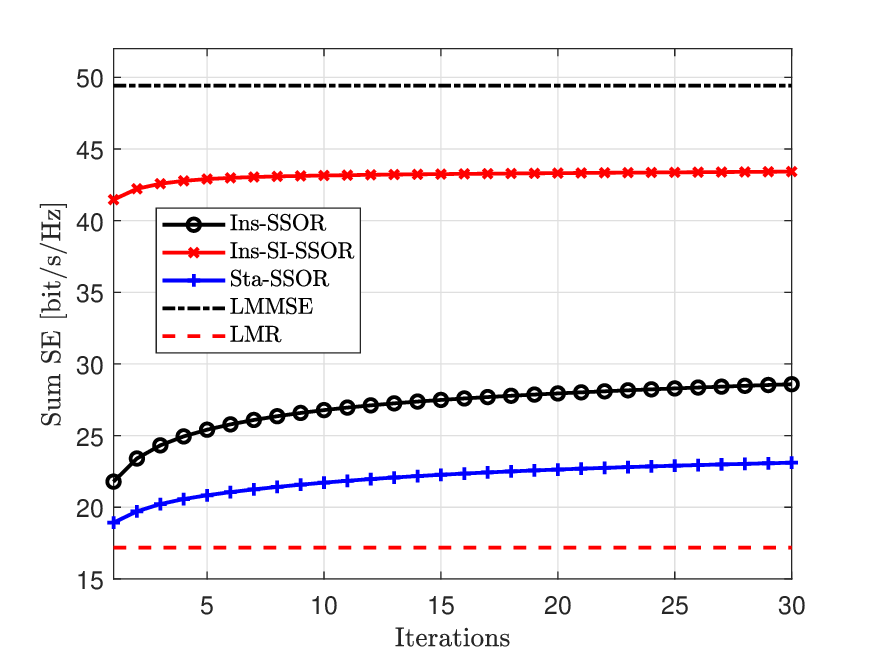}
\caption{Convergence examples of the SSOR algorithm-based combining schemes with $M=8$, $N_x=8$, $K=10$, and $\Delta_x=\lambda/8$. \label{FIG_7_SSOR_Cover}}
\vspace{-0.3cm}
\end{figure}

\begin{figure}[t]
\centering
\includegraphics[scale=0.5]{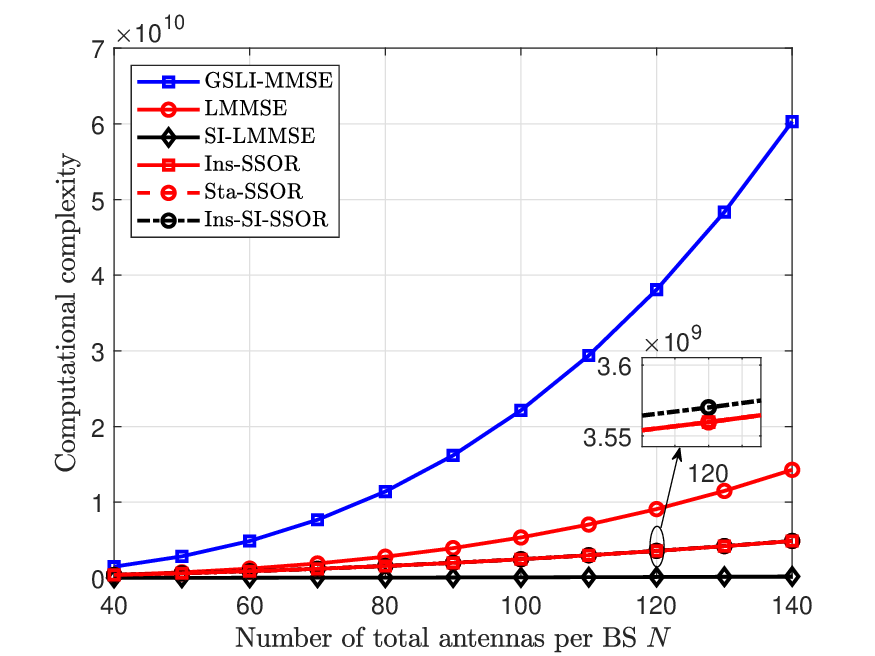}
\caption{Computational complexity of all studied combining schemes with $M=6$, $K=10$, $N_{Iter}=5$, and $N_r=800$. \label{FIG_8_Complexity}}
\vspace{-0.3cm}
\end{figure}

\begin{figure}[t]\centering
\vspace{0.3cm}
\subfigure[MMSE channel estimator]{
\begin{minipage}{8cm}\centering
\includegraphics[scale=0.5]{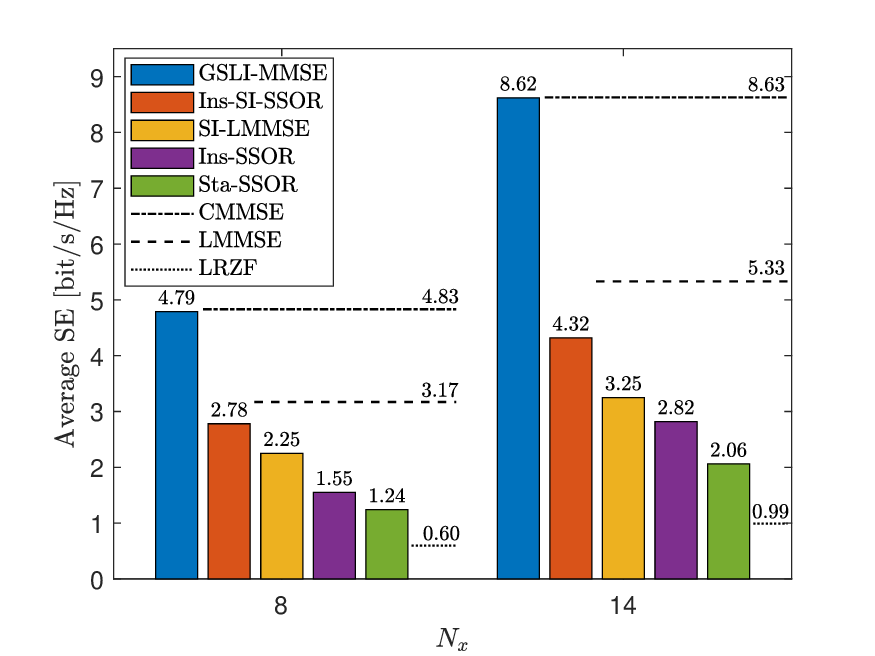}
\end{minipage}}
\subfigure[EW-MMSE channel estimator]{
\begin{minipage}{8cm}\centering
\includegraphics[scale=0.5]{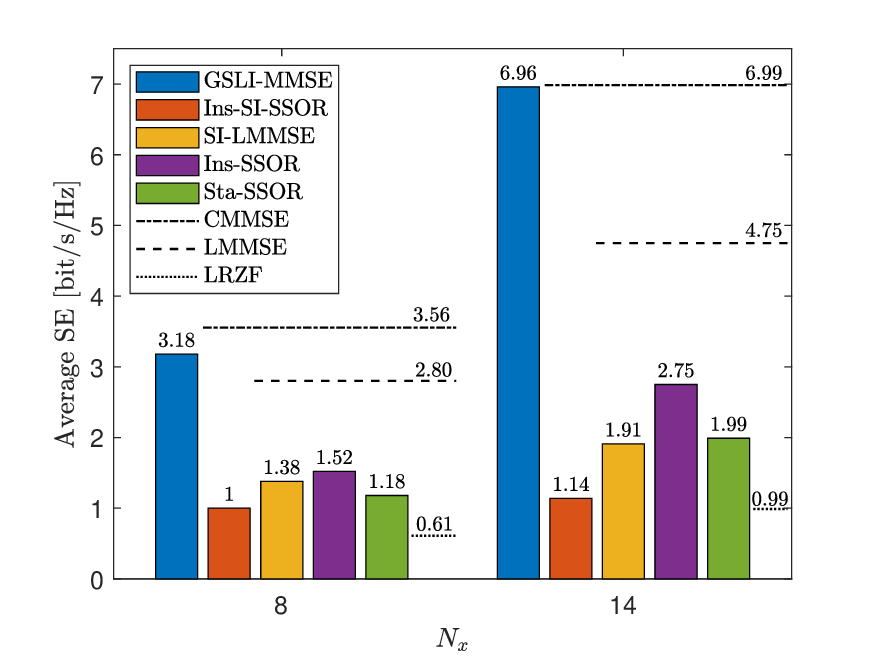}
\end{minipage}}
\caption{Average SE performance for all studied combining schemes over the MMSE and EW-MMSE channel estimators with $M=4$, $K=10$, and $\Delta_x=\lambda/8$. \label{FIG_9_All_Schemes_Pillar}}
\end{figure}

Fig.~\ref{FIG_7_SSOR_Cover} showcases the convergence examples for three SSOR algorithm-based combining schemes. As observed, with the aid of the initial values, the Ins-SI-SSOR can significantly outperform the other two SSOR algorithm-based combining schemes, whether from the perspective of the achievable performance or the convergence speed. More specifically, for instance, $20.24 \%$ and $49.04 \%$ SE improvement can be achieved by the Ins-SI-SSOR scheme compared with the Ins-SSOR and Sta-SSOR schemes, respectively, for $N_{Iter}=10$. Besides, the Ins-SI-SSOR scheme converges very fast, where only very few iterations are needed to achieve the excellent performance.

To have a comprehensive vision of all considered combining schemes in this paper, we now compare all studied schemes in this paper from the perspectives of the computational complexity and the achievable performance. Fig.~\ref{FIG_8_Complexity} compares the computational complexity of all proposed combining schemes against the total number of antennas each BS $N$. As observed, the GSLI-MMSE combining scheme involves the largest computational complexity, where the main computational complexity comes from the precomputation based on the statistics information as illustrated in Table~\ref{Comparisons}. Meanwhile, the SI-LMMSE combining scheme showcases the significantly smaller computational complexity than other schemes, since the statistics matrix inversion is only needed to be computed once in each realization of the BS/UE locations among all $N_r$ channel realizations. Besides, the SSOR algorithm-based schemes embrace much smaller computational complexity than the LMMSE combining scheme. Since the forward order and reverse order computations as in \eqref{SSOR_1} and \eqref{SSOR_2}, respectively, are required in every iteration and every channel realization, the computational complexity of the SSOR algorithm-based schemes is larger than that of the SI-LMMSE scheme.

Finally, in Fig.~\ref{FIG_9_All_Schemes_Pillar}, we compare the average SE performance for all studied combining schemes in this paper over MMSE and EW-MMSE channel estimators. We consider CMMSE, LMMSE, and local regularized zero-forcing (LRZF) combining as \cite[Eq. (4.9)]{8187178} as the benchmarks. Note that the LRZF combining is regarded as a low-complexity approach derived from \eqref{eq_LMMSE} by neglecting the covariance matrix-related terms in the matrix inversion and applying the standard matrix inversion result. As clearly observed in Fig.~\ref{FIG_9_All_Schemes_Pillar}, our studied combining schemes can achieve excellent SE performance compared to the benchmarks. More specifically, the GSLI-MMSE combining scheme can achieve the approaching performance to that of the CMMSE combining scheme over both the MMSE and EW-MMSE channel estimators. Note that the SSOR-based combining schemes and the SI-LMMSE combining scheme perform better over the MMSE channel estimator than the EW-MMSE channel estimator. For instance, with $N_x=8$, the performance gaps between the Ins-SI-SSOR combining scheme and the LMMSE combining scheme over the MMSE and EW-MMSE channel estimators are $12.3 \%$ and $64.29\%$, respectively. This observation demonstrates the importance of deriving the accurate estimation of the channel covariance matrices to implement low-complexity combining schemes to achieve excellent SE performance. Moreover, all these schemes can significantly outperform the lower benchmark generated by the LRZF combining scheme.

\section{Conclusions}\label{conclusion}
Low-complexity distributed combining scheme design for near-field CF XL-MIMO has been studied in this paper. Firstly, we have obtained the uplink SE performance analysis frameworks for CF XL-MIMO networks with the centralized and distributed processing schemes. We have derived the CMMSE and LMMSE combining schemes, which held for arbitrary channel estimators. Then, we have applied the matrix approximation methodologies to derive the GSLI-MMSE and SI-LMMSE combining schemes, where the global instantaneous information in the CMMSE combining and the local instantaneous information in the LMMSE combining is approximated by the global and local statistics information, respectively. Furthermore, we have derived three distributed SSOR-based LMMSE combining schemes by applying the low-complexity SSOR algorithm: Ins-SSOR-LMMSE, Sta-SSOR-LMMSE, and Ins-SI-SSOR-LMMSE combining schemes, which have been distinguished from the applied information and initial values. Numerical results have demonstrated that proposed low-complexity combining schemes can achieve the excellent SE performance with much smaller computational complexity than the conventional MMSE combining schemes.

\begin{appendices}
\section{Some Useful Lemmas}\label{useful}
\begin{lemm}\label{matrixinversion}
For the matrices $\mathbf{A}\in \mathbb{C} ^{n_1\times n_1}$, $\mathbf{B}\in \mathbb{C} ^{n_1\times n_2}$, and $\mathbf{C}\in \mathbb{C} ^{n_1\times n_2}$, we have 
\begin{equation}
( \mathbf{A}+\mathbf{BCB}^H ) ^{-1}\mathbf{B}=\mathbf{A}^{-1}\mathbf{B}( \mathbf{B}^H\mathbf{A}^{-1}\mathbf{B}+\mathbf{C}^{-1} ) ^{-1}\mathbf{C}^{-1},
\end{equation}
which can be easily proved based on the standard matrix inversion lemma as in \cite[Lemma B.3]{8187178}.
\end{lemm}
\begin{lemm}\label{asymptotic}
For $\mathbf{A}\in \mathbb{C} ^{n_1\times n_1}$, $\mathbf{x}\sim \mathcal{N} _{\mathbb{C}}( \mathbf{0},\frac{1}{n_1}\mathbf{I}_{n_1} ) $, $\mathbf{y}\sim \mathcal{N} _{\mathbb{C}}( \mathbf{0},\frac{1}{n_1}\mathbf{I}_{n_1} ) $,  $\mathbf{x}$ and $\mathbf{y}$ are assumed to be mutually independent and independent of $\mathbf{A}$. Moreover, we assume that $\mathbf{A}$ showcases uniformly bounded spectral norm. Following the standard results in \cite{hoydis2013massive}, we have
\begin{align}
&\mathbf{x}^H\mathbf{Ax}\asymp \frac{1}{n_1}\mathrm{tr}\left( \mathbf{A} \right), \label{asy1} \\
&\mathbf{x}^H\mathbf{Ay}\asymp 0.\label{asy2}
\end{align}
\end{lemm}
\vspace{-0.5cm}
% \begin{lemm}\label{lemmatrace}
% Considering arbitrary vectors $\mathbf{x},\mathbf{y}\in \mathbb{C} ^{n_1}$ and the matrix $\mathbf{A}\in \mathbb{C} ^{n_1\times n_1}$, we have
% \begin{equation}\label{trace}
% \mathbf{x}^H\mathbf{A}\mathbf{y}=\mathrm{tr}(\mathbf{A}\mathbf{y}\mathbf{x}^H).
% \end{equation}
% \end{lemm}

\begin{figure*}
{{\begin{align}
&R_{BS,n_1n_{1}^{\prime}}^{ab}=-\frac{\eta}{8\pi}\cos \left[ u_0\left( d_{ab,v} \right) \right] \left\{ \begin{array}{c}
	-2C_i\left[ u_1\left( d_{ab,v},d_{n_1n_{1}^{\prime},h} \right) \right] -2C_i\left[ u_{1}^{\prime}\left( d_{ab,v},d_{n_1n_{1}^{\prime},h} \right) \right] +C_i\left[ u_2\left( d_{ab,v},d_{n_1n_{1}^{\prime},h} \right) \right]\notag\\ 
	+C_i\left[ u_{2}^{\prime}\left( d_{ab,v},d_{n_1n_{1}^{\prime},h} \right) \right] +C_i\left[ u_3\left( d_{ab,v},d_{n_1n_{1}^{\prime},h} \right) \right] +C_i\left[ u_{3}^{\prime}\left( d_{ab,v},d_{n_1n_{1}^{\prime},h} \right) \right]  \notag\\
\end{array} \right\} 
\\
&+\frac{\eta}{8\pi}\sin \left[ u_0\left( d_{ab,v} \right) \right] \left\{ \begin{array}{c}
	2S_i\left[ u_1\left( d_{ab,v},d_{n_1n_{1}^{\prime},h} \right) \right] -2S_i\left[ u_{1}^{\prime}\left( d_{ab,v},d_{n_1n_{1}^{\prime},h} \right) \right] -S_i\left[ u_2\left( d_{ab,v},d_{n_1n_{1}^{\prime},h} \right) \right] \notag \\ 
	+S_i\left[ u_{2}^{\prime}\left( d_{ab,v},d_{n_1n_{1}^{\prime},h} \right) \right] -S_i\left[ u_3\left( d_{ab,v},d_{n_1n_{1}^{\prime},h} \right) \right] +S_i\left[ u_{3}^{\prime}\left( d_{ab,v},d_{n_1n_{1}^{\prime},h} \right) \right] \tag{40} \label{parallel_r}\\ 
\end{array} \right\} 
\\
&X_{BS,n_1n_{1}^{\prime}}^{ab}=-\frac{\eta}{8\pi}\cos \left[ u_0\left( d_{ab,v} \right) \right] \left\{ \begin{array}{c}
	2S_i\left[ u_1\left( d_{ab,v},d_{n_1n_{1}^{\prime},h} \right) \right] +2S_i\left[ u_{1}^{\prime}\left( d_{ab,v},d_{n_1n_{1}^{\prime},h} \right) \right] -S_i\left[ u_2\left( d_{ab,v},d_{n_1n_{1}^{\prime},h} \right) \right] \notag\\
	-S_i\left[ u_{2}^{\prime}\left( d_{ab,v},d_{n_1n_{1}^{\prime},h} \right) \right] -S_i\left[ u_3\left( d_{ab,v},d_{n_1n_{1}^{\prime},h} \right) \right] -S_i\left[ u_{3}^{\prime}\left( d_{ab,v},d_{n_1n_{1}^{\prime},h} \right) \right] \notag\\ 
\end{array} \right\} 
\\
&+\frac{\eta}{8\pi}\sin \left[ u_0\left( d_{ab,v} \right) \right] \left\{ \begin{array}{c}
	2C_i\left[ u_1\left( d_{ab,v},d_{n_1n_{1}^{\prime},h} \right) \right] -2C_i\left[ u_{1}^{\prime}\left( d_{ab,v},d_{n_1n_{1}^{\prime},h} \right) \right] -C_i\left[ u_2\left( d_{ab,v},d_{n_1n_{1}^{\prime},h} \right) \right] \notag\\
	+C_i\left[ u_{2}^{\prime}\left( d_{ab,v},d_{n_1n_{1}^{\prime},h} \right) \right] -C_i\left[ u_3\left( d_{ab,v},d_{n_1n_{1}^{\prime},h} \right) \right] +C_i\left[ u_{3}^{\prime}\left( d_{ab,v},d_{n_1n_{1}^{\prime},h} \right) \right] \tag{41} \label{parallel_x}\\
\end{array} \right\} 
\end{align}}
\hrulefill
\vspace*{-0.6cm}
%\vspace*{3pt}
}\end{figure*}

\section{Modeling of $\mathbf{Z}_{BS,C}$ in \eqref{MC_Matrix}}\label{ZBS}

As for $\mathbf{Z}_{BS,C}$, for the array configuration as in Fig.~\ref{System}, motivated by \cite{wang2024analytical}, $\mathbf{Z}_{BS,C}$ can be constructed by $N_y\times N_y$ sub-matrices with the dimension of $N_x\times N_x$ each, that is $\mathbf{Z}_{BS,C}=[ \mathbf{Z}_{BS,C}^{ab} ] _{N_y\times N_y}$ with $\mathbf{Z}_{BS,C}^{ab} \in \mathbb{C} ^{N_x\times N_x} $ being the mutual impedance matrix between the $a$-th row of antennas and the $b$-th row of antennas, where  $\{ a,b \} =\{ 1,\dots ,N_y \}$. The mutual impedance between the $n_1$-th antenna in the $a$-th row and the $n_{1}^{\prime}$-th antenna in the $b$-th row, counting from the left to right, can be denoted as $z_{BS,n_1n_{1}^{\prime}}^{ab}$, which is the $(n_1,n_{1}^{\prime})$-th element of $\mathbf{Z}_{BS,C}^{ab} $. We can uniformly formulate $z_{BS,n_1n_{1}^{\prime}}^{ab}$ as $z_{BS,n_1n_{1}^{\prime}}^{ab}=R_{BS,n_1n_{1}^{\prime}}^{ab}+jX_{BS,n_1n_{1}^{\prime}}^{ab}$ and we denote $d_{n_1n_{1}^{\prime},h}=| n_1-n_{1}^{\prime} | \Delta _{x}$, $d_{ab,v}=| a-b | \Delta _{y}$, and $d_{n_1n_{1}^{\prime}}^{ab}=\sqrt{( n_1-n_{1}^{\prime} ) ^2\Delta _{x}^{2}+( a-b ) ^2\Delta _{y}^{2}}$ as the horizontal distance, vertical distance, and distance between the $n_1$-th antenna in the $p$-th row and the $n_{1}^{\prime}$-th antenna in the $q$-th row, respectively. Based on \cite[Sec. 8.6.2]{balanis2016antenna}, we can compute $z_{BS,n_1n_{1}^{\prime}}^{ab}$ over four possible combinations for $p,q$, and $n_1,n_{1}^{\prime}$. For $a=b, n_1=n_{1}^{\prime}$, $z_{BS,n_1n_{1}^{\prime}}^{ab}=Z_A$. For $a=b, n_1\ne n_{1}^{\prime}$, according to the ``side-by-side" scenario in \cite[Sec. 8.6.2]{balanis2016antenna}, we have $R_{BS,n_1n_{1}^{\prime}}^{ab}=\frac{\eta}{4\pi}\{ 2C_i[ v_0( d_{n_1n_{1}^{\prime},h} ) ] -C_i[ v_1( d_{n_1n_{1}^{\prime},h} ) ] -C_i[ v_2( d_{n_1n_{1}^{\prime},h} ) ] \} $ and $X_{BS,n_1n_{1}^{\prime}}^{ab}=-\frac{\eta}{4\pi}\{ 2S_i[ v_0( d_{ab,v} ) ] -S_i[ v_1( d_{n_1n_{1}^{\prime},h} ) ] -S_i[ v_2( d_{n_1n_{1}^{\prime},h} ) ] \} $
with $v_0( d_{n_1n_{1}^{\prime},h} ) =\kappa d_{n_1n_{1}^{\prime},h}$, $v_1( d_{n_1n_{1}^{\prime},h} ) =\kappa [ \sqrt{d_{n_1n_{1}^{\prime},h}^{2}+\Delta_{l}^{2}}+\Delta _l ] $, and $v_2( d_{n_1n_{1}^{\prime},h} ) =\kappa [ \sqrt{d_{n_1n_{1}^{\prime},h}^{2}+\Delta_{l}^{2}}-\Delta _l ] $. For $a\ne b, n_1=n_{1}^{\prime}$, according to the 
in \eqref{parallel_r} and \eqref{parallel_x}, respectively, where $u_0( d_{ab,v} ) =\kappa d_{ab,v}$, $u_1( d_{ab,v},d_{n_1n_{1}^{\prime},h} ) =\kappa [ ( d_{n_1n_{1}^{\prime},h}^{2}+d_{ab,v}^{2} ) ^{1/2}+d_{ab,v} ] $, $u_{1}^{\prime}( d_{ab,v},d_{n_1n_{1}^{\prime},h} ) =\kappa [ ( d_{n_1n_{1}^{\prime},h}^{2}+d_{ab,v}^{2} ) ^{1/2}-d_{ab,v} ] $, $u_2( d_{ab,v},d_{n_1n_{1}^{\prime},h} ) =\kappa \{ [ d_{n_1n_{1}^{\prime},h}^{2}+( d_{ab,v}-\Delta _l ) ^2 ] ^{1/2}+( d_{ab,v}-\Delta _l ) \} $, $u_{2}^{\prime}( d_{ab,v},d_{n_1n_{1}^{\prime},h} ) =\kappa \{ [ d_{n_1n_{1}^{\prime},h}^{2}+( d_{ab,v}-\Delta _l ) ^2 ] ^{1/2}-( d_{ab,v}-\Delta _l ) \} $, $u_3( d_{ab,v},d_{n_1n_{1}^{\prime},h} ) =\kappa \{ [ d_{n_1n_{1}^{\prime},h}^{2}+( d_{ab,v}+\Delta _l ) ^2 ] ^{1/2}+( d_{ab,v}-\Delta _l ) \} $, and $u_{3}^{\prime}( d_{ab,v},d_{n_1n_{1}^{\prime},h} ) =\kappa \{ [ d_{n_1n_{1}^{\prime},h}^{2}+( d_{ab,v}+\Delta _l ) ^2 ] ^{1/2}-( d_{ab,v}-\Delta _l ) \} $, respectively. Moreover, for the final possible scenario $a\ne b, n_1=n_{1}^{\prime}$, according to the ``collinear" scenario in \cite[Sec. 8.6.2]{balanis2016antenna}, we have 
\setcounter{equation}{41}
\begin{equation}
\begin{aligned}
R_{BS,n_1n_1}^{ab}&=-\frac{\eta}{8\pi}\cos [ u_0( d_{ab,v} ) ] \{ -2C_i[ 2u_0( d_{ab,v} ) ] +C_i[ u_5( d_{ab,v} ) ] \\
&+C_i [ u_4 ( d_{ab,v}  )  ] -\ln  [ u_6 ( d_{ab,v}  )  ]  \} +\frac{\eta}{8\pi}\sin  [ u_0 ( d_{ab,v}  )  ]  \\
&\{ 2S_i [ 2u_0 ( d_{ab,v}  )  ] -S_i [ u_5 ( d_{ab,v}  )  ] S_i [ u_4 ( d_{ab,v}  )  ]  \},
\end{aligned}
\end{equation}
\begin{equation}
\begin{aligned}
X_{BS,n_1n_1}^{ab}&=-\frac{\eta}{8\pi}\cos [ u_0( d_{ab,v}  )  ] \{ 2S_i[ 2u_0( d_{ab,v}  )  ] -S_i[ u_5( d_{ab,v}  )  ] \\
&+S_i[ u_4( d_{ab,v}  )  ]  \} +\frac{\eta}{8\pi}\sin [ u_0( d_{ab,v}  )  ] \{ 2C_i[ 2u_0( d_{ab,v}  )  ] \\
&-C_i[ u_5( d_{ab,v}  )  ] -C_i[ u_4( d_{ab,v}  )  ] -\ln [ u_6( d_{ab,v}  )  ]  \}
\end{aligned}
\end{equation}
where $u_4( d_{ab,v} ) =2\kappa ( d_{ab,v}+\Delta _l ) $ and $u_5( d_{ab,v} ) =2\kappa ( d_{ab,v}-\Delta _l ) $. In summary, we can easily derive $\mathbf{Z}_{BS}$ by plugging all these results in \eqref{MC_Matrix}.

\end{appendices}

\bibliographystyle{IEEEtran}
\bibliography{IEEEabrv,Ref}

\end{document}